\documentclass[a4paper,11pt]{article}

\usepackage[utf8]{inputenc}
\usepackage[affil-it]{authblk}
\usepackage{amsmath}
\usepackage{amssymb}
\usepackage{pifont}
\usepackage[letterpaper,top=2cm,bottom=2cm,left=3cm,right=3cm,marginparwidth=1.75cm]{geometry}
\usepackage{graphicx}
\usepackage[colorlinks=true, allcolors=blue]{hyperref}
\usepackage{amsthm}
\usepackage[nottoc]{tocbibind}
\usepackage{braket}

\usepackage{tikz} 
\usetikzlibrary{arrows, shapes.gates.logic.US, calc}
\usetikzlibrary{quantikz}

\usepackage{verbatim}

\newtheorem{theorem}{Theorem}

\newtheorem{lemma}{Lemma}
\newtheorem{corollary}{Corollary}

\newtheorem{definition}{Definition}

\newcommand{\setft}[1]{\mathrm{#1}}
\newcommand{\Density}{\setft{D}}
\newcommand{\Pos}{\setft{Pos}}

\newcommand{\Lin}{\setft{L}}

\DeclareMathOperator{\negl}{negl}

\newcommand{\Tr}{\mbox{\rm Tr}}

\newcommand{\mZ}{\mathbb{Z}}

\newcommand{\mH}{\mathcal{H}}

\newcommand{\supp}{\textsc{Supp}}

\newcommand{\lwe}{$\mathrm{LWE}$}
\newcommand{\LWE}{$\mathrm{LWE}$}
\newcommand{\LPN}{$\mathrm{LPN}$}

\newcommand{\SZK}{\ensuremath{\sf SZK}}
\newcommand{\SZKf}{\ensuremath{\sf SZK_{\ensuremath{\delta}}}}
\newcommand{\SZKlog}{\ensuremath{\sf SZK_{log}}}
\newcommand{\SZKlogd}{\ensuremath{\sf SZK_{polylog}}}
\newcommand{\SZKconst}{\ensuremath{\sf SZK_{const}}}
\newcommand{\QSZK}{\ensuremath{\sf QSZK}}
\newcommand{\QMA}{\ensuremath{\sf QMA}}
\newcommand{\QIP}{\ensuremath{\sf QIP}}
\newcommand{\QSZKf}{\ensuremath{\sf QSZK_{\ensuremath{\delta}}}}
\newcommand{\QSZKlog}{\ensuremath{\sf QSZK_{log}}}
\newcommand{\QSZKlogd}{\ensuremath{\sf QSZK_{polylog}}}
\newcommand{\QSZKconst}{\ensuremath{\sf QSZK_{const}}}
\newcommand{\BQP}{\ensuremath{\sf BQP}}
\newcommand{\BPP}{\ensuremath{\sf BPP}}
\newcommand{\NC}{\ensuremath{\sf NC^1}}
\newcommand{\NCz}{\ensuremath{\sf NC^0}}
\newcommand{\QNC}{\ensuremath{\sf QNC^1}}
\newcommand{\QNCz}{\ensuremath{\sf QNC^0}}
\newcommand{\QNCf}{\ensuremath{\sf QNC_{f}^0}}

\newcommand{\QCD}{\textsc{QCD}}
\newcommand{\QED}{\textsc{QED}}
\newcommand{\QEDlog}{QED$_{\textrm{log}}$}
\newcommand{\QEDlogd}{QED$_{\textrm{polylog}}$}
\newcommand{\QEDconst}{QED$_{\textrm{O(1)}}$}
\newcommand{\ED}{ED}
\newcommand{\EDlog}{ED$_{\textrm{log}}$}
\newcommand{\EDlogd}{ED$_{\textrm{polylog}}$}
\newcommand{\EDconst}{ED$_{\textrm{O(1)}}$}

\newcommand{\QER}{\textsc{QER}}

\newcommand{\QERconst}{QER$_{\textrm{O(1)}}$}

\newcommand{\EDf}{ED$_{\delta}$}
\newcommand{\QEDf}{QED$_{\delta}$}

\newcommand{\HQED}{\textsc{HQED}}
\newcommand{\HQEDlog}{HQED$_{\textrm{log}}$}
\newcommand{\HQEDconst}{HQED$_{\textrm{O(1)}}$}

\newcommand{\COf}{\ensuremath{\ensuremath{C}^{\mathcal{O}(f)}}}
\newcommand{\COg}{\ensuremath{\ensuremath{C}^{\mathcal{O}(g)}}}

\newcommand{\HQEDb}{HQED$^{b}_{\textrm{O(1)}}$}

\newif\ifnotes\notestrue

\newcommand{\mnote}[1]{\textcolor{red}{\small {\textbf{(Matty:} #1\textbf{)
      }}}}
\newcommand{\gnote}[1]{\textcolor{blue}{\small {\textbf{(Andru:} #1\textbf{) }}}}

\ifnotes
\usepackage{color}
\definecolor{mygrey}{gray}{0.50}
\newcommand{\notename}[2]{{\textcolor{mygrey}{\footnotesize{\bf (#1:} {#2}{\bf ) }}}}

\newcommand{\pnote}[1]{{\endnote{#1}}}

\else

\newcommand{\notename}[2]{{}}

\newcommand{\pnote}[1]{}

\renewcommand{\mnote}[1]{}
\renewcommand{\gnote}[1]{}

\fi

\begin{document}
\sloppy
\title{On estimating the entropy of shallow circuit outputs}
\author{Alexandru Gheorghiu$^{1}$ and Matty J. Hoban$^{2}$}

\footnotetext[1]{Department of Computer Science and Engineering, Chalmers University of Technology.}
\footnotetext[2]{Department of Computer Science, University of Oxford.}

\date{}
\maketitle
\sloppy

\begin{abstract}
\noindent
Estimating the entropy of probability distributions and quantum states is a fundamental task in information processing. Here, we examine the hardness of this task for the case of probability distributions or quantum states produced by shallow circuits. Specifically, we show that entropy estimation for distributions or states produced by either log-depth circuits or constant-depth circuits with gates of bounded fan-in and \emph{unbounded} fan-out is at least as hard as the \emph{Learning with Errors} (LWE) problem, and thus believed to be intractable even for efficient quantum computation. 
This illustrates that quantum circuits do not need to be complex to render the computation of entropy a difficult task. We also give complexity-theoretic evidence that this problem for log-depth circuits is not as hard as its counterpart with general polynomial-size circuits, seemingly occupying an intermediate hardness regime. Finally, we discuss potential future applications of our work for quantum gravity research by relating our results to the complexity of the bulk-to-boundary dictionary of AdS/CFT.
\end{abstract}

\newpage

\tableofcontents

\newpage

\section{Introduction} \label{sect:intro}

The entropy of a probability distribution or a quantum state is a useful measure for characterising information content with numerous applications in information theory. Shannon and von Neumann entropies, both appear in data compression \cite{schumacher} and asymptotic cryptography \cite{qkd}, and in the case of the von Neumann entropy, entanglement theory \cite{entanglement}. Furthermore, this link to entanglement theory has led to the use of the von Neumann entropy in condensed matter theory \cite{arealaws} and quantum gravity research \cite{rt}.

Given its importance in physics and information theory, it is natural to ask how difficult it is to estimate the entropy of a process. 
A natural way to formalise this question is in terms of \emph{sample complexity}: this looks at how many samples from an unknown probability distribution, or how many copies of an unknown quantum state, are needed to estimate its entropy.
Previous work has considered various algorithms, both quantum and classical, for  computing the entropy of a probability distribution \cite{dasgupta2002complexity, wuyang, jiaovenkat, valiantvaliant}, as well as entropies of quantum states \cite{hastings, acharya, quantumalgorithmentropy}. More recently, it has been shown that for multiple entropy measures, computing the relevant entropy is as hard as full state tomography \cite{acharya, quantumalgorithmentropy, o2015quantum}. In other words, the sample complexity would scale with the support of the probability distribution, or the dimension of the Hilbert space for the quantum state, respectively.
To provide some intuition for why entropy estimation is difficult, in Appendix~\ref{sect:bqpqpoly}, we use the computational complexity tool of advice to give a simple proof that no algorithm exists for efficiently estimating the entropy of an a priori unknown quantum state. In particular, we show that if the entropy of a (mixed) quantum state could be estimated within additive error $\epsilon$ in time that scales polynomially in the number of qubits of the state and in $\log(1/\epsilon)$, then such an algorithm could be leveraged to solve \emph{any} problem. To be more precise, it would imply that polynomial-time quantum computation with quantum advice, could decide all languages, which is known to be false.

Rather than considering sample complexity, an operational way of capturing the complexity of entropy estimation is to start with descriptions of two random processes and ask which process produces more entropy in its output. The natural decision problem for this task was defined by Goldreich and Vadhan~\cite{goldreich} and is known as the \emph{entropy difference} (\ED) problem: given two classical circuits $C_1$ and $C_2$, let $C_1(x)$ and $C_2(x)$ define the output distributions of the two circuits when acting on an $n$-bit string, $x$, drawn from the uniform distribution over $\{0, 1\}^n$; the problem is to decide whether $C_1(x)$ has higher entropy than $C_2(x)$ or vice versa (promised that one of these is the case and promised that the entropy difference is at least $1/poly(n)$).

What can one say about the computational hardness of this problem? In~\cite{goldreich} it was shown that the problem is complete for the class \SZK. This class, known as \emph{statistical zero-knowledge}, contains all decision problems for which a computationally unbounded \emph{prover} can convince a polynomial-time \emph{verifier} that there exists a solution, when one exists (and fail to do so when a solution does not exist) without revealing anything about the solution to the verifier\footnote{This last condition, known as the zero-knowledge condition, is formally defined by saying that there exists a polynomial-sized circuit called a \emph{simulator} that can approximately produce transcripts of the protocol for accepting (or ``yes'') instances.}.
Due to oracle separation results~\cite{szk1} and from the fact that certain cryptographic tasks (such as finding collisions for a cryptographic hash function) are contained in \SZK, it is believed that \SZK\ contains problems that are intractable even for polynomial-time quantum algorithms.
Thus, the fact that \ED\ is complete for \SZK\ tells us that we should expect a similar intractability for the problem of distinguishing the entropies of general classical circuits.

Transitioning to the quantum setting, Ben-Aroya, Schwartz and Ta-Shma defined the analogue of \ED\ known as the \emph{quantum entropy difference} (\QED) problem~\cite{bst}. In this case, the input circuits $C_1$ and $C_2$ are polynomial-size quantum circuits acting on a fixed input state (say the state $\ket{00...0}$), and a fraction of the output qubits are traced out. The remaining qubits will generally be mixed quantum states. As in the classical case, the question is which of the two mixed states (the one produced by $C_1$ or the one produced by $C_2$) has higher entropy, subject to the same promise as for \ED\ (that there is a $1/poly(n)$ gap in the entropy difference). Ben-Aroya et al showed that \QED\ is complete for the class \QSZK, the quantum counterpart of \SZK\ of problems admitting a \emph{quantum statistical zero-knowledge} proof protocol~\cite{bst}. Assuming \QSZK\ strictly contains \SZK, this would mean that \QED\ is strictly harder than \ED.

For both \ED\ and \QED\ the circuits under consideration were assumed to be polynomial in the size of their inputs. A natural follow-up question is: does the hardness of these problems change if we reduce the depth of the circuits? Specifically, what happens if the circuits have depth that is logarithmic or even constant in the size of the input? Focusing specifically on the quantum case, there are number of reasons why one would be interested in answering these questions. From a complexity-theory perspective, it lets us compare and contrast \QSZK\ to other classes in the interactive proofs model, such as \QMA\ or \QIP. Both of these classes have natural complete problems with the inputs being circuits and in both cases the problems remain complete if the circuits are made to have logarithmic depth~\cite{rosgenQIP, jiqmalog}. Furthermore, from a more practical perspective, given that current and near-term quantum computing devices are subject to noise and imperfections, it is expected that the states produced in these experiments will be the result of circuits of low depth.
Estimating the entropy of these states would help in computing other quantities of interest such as the amount of entanglement~\cite{entanglement} or the Gibbs free energy~\cite{freenergy}. It is therefore important to know whether entropy estimation can be performed efficiently for states produced by shallow circuits.



Lastly, we are motivated to answer these questions by recent connections between quantum gravity research and quantum information theory, in the form of the AdS/CFT correspondence~\cite{adscft}. Briefly, AdS/CFT is a correspondence between a quantum gravity theory in a hyperbolic \emph{bulk} space-time known as \emph{Anti-de Sitter} (AdS) space and a \emph{conformal field theory} (CFT) on the \emph{boundary} of that space-time. The general idea is to compute physical quantities of interest in the bulk quantum gravity theory by mapping them to the boundary field theory. A surprising result to come out of this program is a correspondence between bulk geometry and boundary entanglement known as the \emph{Ryu-Takayanagi formula}~\cite{rt}. It states that, to leading order, the area of a bulk surface is equal to the entropy of the reduced state on the part of the boundary that encloses that surface.
Moreover, for certain families of boundary field theories, it is conjectured that the underlying states can be approximated by logarithmic depth tensor networks known as MERA (\emph{multi-scale entanglement renormalization ansatz})~\cite{swingle}. 
Thus, characterising the complexity of entropy difference for shallow circuits could yield insight into the complexity of distinguishing different bulk geometries in quantum gravity.

Intuitively, polynomial-depth quantum circuits can produce highly entangled quantum states and so their reduced density matrices can be highly mixed. In general it is difficult to distinguish two very mixed quantum states and furthermore it is difficult to distinguish their entropies, even if they only differ by a constant. In some sense, the local quantum information can be ``scrambled" in such circuits. However, in shallow quantum circuits, especially those with constant depth, the situation is more subtle; for one thing, we need to consider whether gates have bounded or unbounded fan-out. In the case of bounded fan-out, Greenberger-Horne-Zeilinger (GHZ) states cannot be generated with constant-depth quantum circuits \cite{watts2019exponential}, but with unbounded fan-out gates GHZ states can be generated \cite{hoyer2005quantum}. In this work, we consider shallow quantum circuits with unbounded fan-out gates.


We initiate the study of entropy distinguishability for shallow circuits (both classical and quantum) with \emph{bounded fan-in and unbounded fan-out}. 
We show that both the classical and quantum versions of entropy difference are hard assuming the intractability of the Learning with Errors (\LWE).
Since \LWE\ serves as a basis for various schemes of \emph{post-quantum cryptography}, our result implies that entropy estimation for shallow circuits is intractable unless these cryptographic schemes are insecure. 

We then show that while it may be hard to estimate the entropy for shallow circuit outputs, both versions of entropy distinguishability are unlikely to be \SZK-complete\ and \QSZK-complete\ respectively, as their analogoues with polynomial-depth circuits are. We do this by showing oracle separations between the problems that can be reduced to shallow circuit entropy distinguishability and their counterparts with general polynomial-depth circuits.
Therefore, the entropy difference problem for both classical and quantum shallow circuits occupies an interesting intermediate complexity regime.
This is in contrast to similar tasks like non-identity testing of quantum circuits (which is $\mathsf{QMA}$-complete), or quantum circuit distinguishability (which is $\mathsf{QIP}$-complete). In those cases, the problems remain $\mathsf{QMA}$-complete and $\mathsf{QIP}$-complete, respectively, for constant depth quantum circuits~\cite{rosgenQIP}.

Finally, we consider a version of von Neumann entropy difference where Hamiltonians are given as input, instead of circuits, and show that this problem is at least as hard as the circuit-based version. This last result allows us to relate these computational complexity results to the simulation of physical systems.
We should also emphasize that all our results apply for generalized R{\'e}nyi entropies as well.

\subsection{Main results}
We start by defining \EDf\ (\QEDf) as the (quantum) entropy difference problem where the circuits are of depth $\delta(n)$, for some monotonically (non-strictly) increasing function $\delta : \mathbb{N} \to \mathbb{N}$, and have gates with bounded fan-in and unbounded fan-out. We will denote \EDlogd\ (\QEDlogd) and \EDconst\ (\QEDconst) to be the collection of all problems \EDf\ (\QEDf), with $\delta(n)$ in $O(polylog(n))$ and $O(1)$, respectively.
Our first result gives an indication that these problems are unlikely to be as hard as their poly-size counterparts.
To show this, we will consider \SZKf\ (\QSZKf) to be the set of problems that reduce to \EDf\ (\QEDf) under polynomial-time reductions\footnote{These correspond to statistical zero-knowledge protocols in which the verifier and simulator circuits have depth $\delta(n)$.}. 
For $\delta(n) = polylog(n)$, we denote these classes as \SZKlogd\ and \QSZKlogd\, respectively, and prove the following:

\begin{theorem}\label{thm1}
There exists an oracle $\mathcal{O}$ such that $\mathsf{SZK^{\mathcal{O}}_{polylog}} \neq \mathsf{SZK}^{\mathcal{O}}$ and $\mathsf{QSZK^{\mathcal{O}}_{polylog}} \neq \mathsf{QSZK}^{\mathcal{O}}$.
\end{theorem}
\noindent In particular, this means that $\mathsf{SZK^{\mathcal{O}}_{log}} \neq \mathsf{SZK}^{\mathcal{O}}$ and $\mathsf{QSZK^{\mathcal{O}}_{log}} \neq \mathsf{QSZK}^{\mathcal{O}}$. The oracle we use is the same as the one from the recent work of Chia, Chung and Lai~\cite{chiachunglai}, showing the separation $\mathsf{BPP}^{\mathsf{QNC}^{\mathcal{O}}} \neq \mathsf{BQP}^{\mathcal{O}}$, where $\mathsf{QNC}$ denotes the set of problems that can be solved by quantum circuits of polylogarithmic depth. In particular, because we use the oracle from~\cite{chiachunglai}, we in fact obtain a hierarchy theorem which says that relative to the oracle $\mathcal{O}$, we have that $(Q)ED^{\mathcal{O}}_{\delta(n)} \leq_P (Q)ED^{\mathcal{O}}_{2\delta(n) + 1}$.
We note that since our initial version of the manuscript, a result of Arora et al.~\cite{arora2022quantum} has shown that $\mathsf{BPP}^{\mathsf{QNC}^{\mathcal{O}}} \neq \mathsf{BQP}^{\mathcal{O}}$ \emph{relative to a random oracle $\mathcal{O}$.} We expect that this extends to our result as well, so that Theorem~\ref{thm1} also holds relative to a random oracle.

Our second result concerns a direct approach at trying to show that \QSZKlog $=$ \QSZK\ and why we believe this is unlikely. To explain this approach, let us first discuss the class \QIP, of problems that can be decided using an interactive proof system having a quantum verifier. A problem that is complete for this class is the quantum circuit distinguishability problem (\QCD), in which one is given as input two quantum circuits and asked to determine whether, when restricting to a subset of input and output qubits, the corresponding channels are close or far in diamond distance. It was shown by Rosgen that this problem remains \QIP-complete even when the circuits under consideration are log-depth (or even constant depth)~\cite{rosgenQIP}.
This is achieved by constructing log-depth circuits that check a ``history state" of the original circuits, which uses a different construction to the Feynman-Kitaev history state~\cite{fk}. One can also show that any \QIP\ protocol can be made to have a log-depth verifier, using the same history state construction.
In analogy to Rosgen's result, we can now suppose that any \QSZK\ protocol can be made into a \QSZKlog\ protocol by having the prover send history states of the computations that the verifier in the \QSZK\ protocol would perform. 
It is clear from the \QIP\ result that making the verifier have logarithmic depth in a \QSZK\ protocol does not reduce the set of problems that can be decided by such a protocol. The question is whether \emph{in addition} to having a log-depth verifier, the transcript of such a protocol\footnote{Here we mean the transcript of the protocol for ``yes'' instances to the problem.} can be produced by a log-depth simulator.
We show that if this is possible, then since \BQP\ $\subseteq$ \QSZK\, it would be possible to simulate any \BQP\ computation by ``checking" this history state on a log-depth quantum computer. This result can be stated as follows:

\begin{theorem}
If there exists a polynomial-time reduction from a \QSZK\ protocol with a log-depth verifier to a \QSZKlog\ protocol which preserves the transcript of the \QSZK\ protocol, then $\BQP=\BPP^{\QNC}$.
\end{theorem}

Based on these results, we conjecture that estimating the output entropy of circuits having (poly)logarithmic or constant depth should be easier than for circuits of larger depth. Despite this fact, our next result shows that even for these shallow circuits, entropy difference is still intractable for polynomial-time classical and quantum algorithms, assuming the quantum intractability of the Learning with Errors (\LWE) problem:

\begin{theorem} \label{thm:lwe}
 \lwe\ $\leq_{P}$ \EDconst\ $\leq_{P}$ \QEDconst.
\end{theorem}
\LWE, a problem defined by Regev in~\cite{lwe}, serves as the basis for recently proposed cryptographic protocols that are believed to be post-quantum secure. In other words, it is conjectured that no polynomial-time classical or quantum algorithm is able to solve this problem. Recent results have leveraged the versatility of this problem to achieve tasks such as verifiable delegation of quantum computation~\cite{mahadev2018classical, gv19}, certifiable randomness generation~\cite{brakerski2018}, and self-testing of a single quantum device~\cite{mv20}. In all of these works, the protocols were based on the use of \emph{Extended Trapdoor Claw-Free Functions} (ETCF), introduced in~\cite{mahadev2018classical}. An ETCF family consists of a pair of one-way functions, $(f, g)$, that are hard to invert assuming the hardness of \LWE. Importantly, $f$ is a $1$-to-$1$ one-way function, whereas $g$ is a $2$-to-$1$ one-way function. An essential property of these functions is that it should be hard (based on \LWE) to determine which is the $1$-to-$1$ function and which is the $2$-to-$1$ function, given descriptions of the functions. This property is known as \emph{injective invariance}. Consider what happens if we evaluate each of these functions on a string $x$ drawn uniformly at random from $\{0, 1\}^n$. Given that $f$ is a $1$-to-$1$ function, the distribution $f(x)$ is still the uniform distribution over $n$-bit strings and will therefore have maximum entropy, $S(f) = n$. On the other hand, since $g$ is a $2$-to-$1$, $g(x)$ will be the uniform distribution on \emph{half} of the strings in $\{0, 1\}^n$, thus having entropy $S(g) = n-1$. This means that given descriptions of $f$ and $g$, if one can distinguish the entropy of their outputs, one can also solve \LWE, which effectively shows that \LWE\ $\leq_{P}$ \ED.

While the above argument shows a reduction from \LWE\ to \ED, to obtain the result of Theorem~\ref{thm:lwe} one would need an ETCF pair of functions that can be evaluated in constant depth (with unbounded fan-out gates). Such a construction is not known to exist and so we adopt a different approach towards proving the result. We start by showing that the ETCF functions defined in~\cite{mahadev2018classical} can be performed using log-depth circuits, thus showing \LWE\ $\leq_{P}$ \EDlog. This follows from the fact that the specific ETCF functions we use involve only matrix and vector multiplications, which can be parallelized and performed in logarithmic depth\footnote{The functions also require the ability to sample from a Gaussian distribution, but this can also be done in logarithmic depth (following a preprocessing step that is done as part of the reduction)~\cite{peikert2010efficient}.}. To then go to constant-depth circuits, we use the result of Applebaum, Ishai and Kushilevitz of compiling log-depth one-way functions to constant-depth~\cite{aik}. The main idea is that, for a given function $f$, that can be evaluated in log-depth, as well as an input string $x$, it is possible to produce a \emph{randomized encoding} of  $f(x)$ using a circuit of constant depth. A randomized encoding is an encoded version of $f(x)$ from which $f(x)$ can be efficiently (and uniquely) reconstructed. The main difficulty of the proof is to show that the constant-depth randomized encoders for $f$ and $g$ remain indistinguishable, based on \LWE. This then implies the desired result, \LWE\ $\leq_{P}$ \EDconst. Since \QEDconst\ is simply the quantum generalization of \EDconst, the reduction \EDconst\ $\leq_{P}$ \QEDconst\ is immediate.

It should be noted that similar ideas are employed in~\cite{dvir2017approximating}, where the authors use randomized encodings to characterise the hardness of estimating the entropy of low-degree polynomials, rather than short-depth circuits. Indeed, the use of randomized encodings to obtain constant-depth circuit implementations is not new, as this is already remarked in~\cite{aik}. However, one important aspect of the constant-depth construction with randomized encodings is that it requires gates with unbounded fan-out. This brings us to a sharp distinction between the classical and the quantum case. While in the classical case it is considered standard (for instance in the definition of $\mathsf{NC^1}$) to have gates with unbounded fan-out, the same is not true for the quantum case. There, due to the unclonability of general quantum states, as well as various physical limitations, the natural choice for representing quantum circuits is with gates having both bounded fan-in and bounded fan-out.

Unfortunately, our results do not extend to the case of constant-depth circuits with bounded fan-out gates. We leave resolving the hardness of that case, for both quantum and classical circuits, as an interesting open problem.

Finally, we wish to investigate the potential application of our results to physics. To make this link, we first define a Hamiltonian version of the quantum entropy difference problem, which we denote \HQED. Informally, instead of being given as input two quantum circuits and being asked to determine which produces a state of higher entropy, our input will consist of the descriptions of two local Hamiltonians $H_1$ and $H_2$. Upon tracing out a certain number of qubits from the ground states of these Hamiltonians, the problem is to decide which of the two reduced states has higher entropy\footnote{This is essentially the same as asking which of the two ground states has higher entanglement entropy, for a given bipartition of the qubits in the two states.}.
Unsurprisingly, it can be shown that this problem is at least at hard as the circuit version: 

\begin{theorem}\label{thmhqed}
\QED\ $\leq_{P}$ \HQED.
\end{theorem}

The proof makes use of the Feynman-Kitaev history state construction to encode the history state of a quantum circuit in the ground state of a Hamiltonian~\cite{fk}. Directly using the history states of the circuits, $C_1$ and $C_2$, from an instance of \QED\ would not necessarily yield the desired result since those states have small overlap with the output state of $C_1$ and $C_2$. To get around this issue, we make use of the padding construction from~\cite{nirkhe_et_al:LIPIcs:2018:9095}, and pad the circuits $C_1$ and $C_2$ with polynomially-many identity gates. It can then be shown that the resulting history states for these new circuits will have overlap $1 - 1/poly(n)$ with the output states of $C_1$ and $C_2$. Correspondingly, up to inverse polynomial factors, the entropy difference between these states and the original ones is preserved.

If we now denote as \HQEDlog\ and \HQEDconst\ the analogous problems in which the ground states can be approximated by circuits of logarithmic and constant depth (with unbounded fan-out gates), respectively, it follows from the previous results that \QEDlog\ $\leq_{P}$ \HQEDlog\ and \QEDconst\ $\leq_{P}$ \HQEDconst, and therefore that \lwe\ $\leq_P$ \HQEDconst\ $\leq_P$ \HQEDlog. Furthermore, if \HQEDb\ is \HQEDconst\ but where the circuits have bounded fan-out then we also have \LPN\ $\leq_P$ \HQEDb. 

We now turn to the question of what impact our results have for the AdS/CFT correspondence. As previously mentioned, if states on the boundary are described by log-depth tensor networks according to MERA, then computing the entropy for these states should inform the bulk geometry contained within a boundary. If computing the entropy of the boundary is difficult even for quantum computers, then this poses a challenge to the quantum version of the Extended Church Turing thesis. That is, in some sense, quantum gravitational systems cannot be simulated efficiently by a quantum computer. Bouland, Fefferman and Vazirani have also explored this challenge to the thesis from the perspective of the wormhole growth paradox using the tools of computational complexity \cite{boulandfeffermanvazirani}. In particular, they showed that one can prepare computationally pseudorandom states on the boundary CFT. These are efficiently preparable ensembles of quantum states that require exponential complexity to distinguish from Haar random states. Importantly, the states can have different circuit complexities and under the AdS/CFT correspondence this would correspond to different wormhole lengths. An observer in the bulk could, in principle, determine the length of a wormhole, however the circuit complexity of the boundary states should be computationally intractable to determine. They conclude that either the map from bulk to boundary in the AdS/CFT duality is exponentially complex, or a quantum computer cannot efficiently simulate such a system. 

With our final result in Theorem \ref{thmhqed} as a basis, we can make tentative connections to the AdS/CFT duality and reach a similar conclusion to the Bouland et al result.
If one can have ground states of differing entanglement entropy, based on \LWE, for a CFT Hamiltonian then this will correspond to differing bulk geometries. An observer in the bulk can, in principle, efficiently differentiate between these geometries. 
Based on this, we devise an experiment that allows for the efficient estimation of the ground state entanglement entropy of our Hamiltonian. We conjecture that unless the AdS/CFT bulk to boundary map is exponentially complex, this experiment would violate the intractability assumption of \LWE. An important point to make is that our results can be applied to any simulation of a system where the AdS/CFT duality holds; it need not be the case that it is a true description of our universe. Thus, our proposed experiment does not rely on an exotic physical system, but can be explored with a quantum simulator.

\subsection{Discussion and open questions}

The current manuscript is substantially different from the original version of our paper. The most important change we made was clarifying that our results apply for circuits with unbounded fan-out gates. While we have made several attempts to characterise the hardness of entropy estimation for constant depth circuits with bounded fan-out gates, these have unfortunately been unsuccessful. As it stands, it's unclear if entropy estimation in that scenario remains hard or becomes tractable (or whether there is a quantum-classical separation in complexity). We therefore leave that case as an exciting open problem.

Let us also comment on a number of follow-up works that came out after version 1 of our paper. The most relevant is~\cite{aaronson2022quantum} which extends our results by introducing the notion of \emph{quantum pseudoentanglement} and connecting it to questions in the study of quantum gravity. Pseudoentangled states are those whose entanglement entropy is computationally intractable to estimate. There is an important distinction in whether one is given merely copies of the pseudoentangled states, or whether one is given descriptions of circuits preparing those states. Our results have focused on the latter, which has been later referred to as \emph{public key quantum pseudoentanglement}~\cite{bouland2023public}. For the former, the authors of~\cite{aaronson2022quantum} showed that (private-key) pseudoentangled states can be constructed assuming only the existence of one-way functions. This contrasts the public key case where both our results and those of~\cite{bouland2023public} require making the more specialised assumption of the hardness of LWE.
Another important distinction between our results and those of~\cite{aaronson2022quantum,bouland2023public}, is that they require entanglement entropy to be hard to estimate across any bipartition of the qubits. In contrast, our results consider only one bipartition of the qubits.
Finally,~\cite{bouland2023public} showed that pseudoentangled states can be made to have a large gap in entanglement entropy (as small as $\Omega(\log n)$ and as large as $O(n)$). This also follows from our result in Appendix~\ref{sect:eratio}, where we examined the entropy ratio between states produced by shallow circuits.\footnote{Though, once again, our result holds for a single bipartition of the qubits, whereas the results in~\cite{bouland2023public} apply for any bipartition.}

A natural open problem raised by our work is to characterise more precisely the complexity classes induced by the shallow-depth versions of entropy difference (\SZKlog, \SZKconst, \QSZKlog, \QSZKconst). We show that these correspond to zero-knowledge protocols in which the verifier and simulator are bounded-depth circuits, but can one give other characterisations of these classes? 
Some work along these lines has already been done in~\cite{dvir2017approximating}. There, the authors considered the entropy estimation problem for polynomial mappings, showing that for low-degree polynomials the resulting problems are complete for $\mathsf{SZK}_L$, the analog of $\mathsf{SZK}$ where the honest verifier and its simulator are computable in logarithmic space.

We have also seen that these depth-restricted classes appear to be \emph{strictly} contained in their counterparts with general polynomial-depth circuits. Thus, if estimating the entropy of shallow circuits is easier, what is the runtime of the optimal quantum algorithm for deciding \QED\ (or \ED) for log-depth and constant-depth circuits?
It is also pertinent to ask this question not just for worst-case instances of the problems, but also for the average case. \lwe\ is assumed to be hard on-average, and thus we can use the reduction to argue that \ED\ should be hard on average for a particular range of parameters. Can we say anything more general? What about the entropy of random quantum circuits? States produced by random circuits will typically be highly entangled and thus subsystems will have close to the maximum amount of entropy. Does this imply that all forms of entropy calculation for random circuits is easy?


Let us also comment on the use of \LWE\ and ETCF functions to derive Theorem~\ref{thm:lwe}. \LWE\ and lattice problems are leading candidates in the development of post-quantum cryptographic protocols due to their conjectured intractability even for quantum algorithms. Moreover, an appealing feature of \LWE\ for cryptographic purposes, is that one can reduce \emph{worst-case} instances of lattice problems to \emph{average-case} instances of \LWE\ \cite{lwe, peikert2016decade}. Combining this with Theorem~\ref{thm:lwe}, means that average-case instances (relative to a certain distribution over input circuits) of \EDconst\ will also be ``at least as hard'' as worst-case lattice problems. Thus, using \LWE\ gives strong evidence that \EDconst\ is quantum intractable.
ETCF functions then provide a very natural instance of an entropy difference problem: determine whether a given function is 1-to-1 or 2-to-1. Such functions differ by one bit in entropy, under a uniform input, but by the injective invariance property it is at least as hard as solving \LWE\ to tell which is which. It is then only necessary, for our results, to show that these functions can be evaluated by circuits of constant depth. As mentioned, this can be achieved in a relatively straightforward manner by first showing that the functions can be evaluated in log depth and then using the compiling technique of Applebaum, Ishai and Kushilevitz to achieve constant depth.

A natural question at this point is whether these ETCF functions were necessary, and could generic one-way functions have sufficed? Furthermore, there are general constructions known as \emph{lossy (trapdoor) functions} \cite{peikert}, which are functions that are either injective or involve many collisions, analogous to our $1$-to-$1$ or $2$-to-$1$ functions. Do these functions not suffice already to give our result? Indeed, our approach gives an example of such a construction of lossy functions. While generic functions could be leveraged to show that there is a separation between the entropies of the two associated distributions, it is not clear \emph{a priori} what the value of this separation is. With our construction it is clear that there is one bit of entropy different between the two distributions, and that our construction can be generalised to give a greater entropy difference. Going further, our construction is adaptable enough to allow us to say something about estimating the entropy of constant-depth quantum circuits with bounded fan-in and bounded fan-out.
Lastly, having depth-efficient constructions for ETCF functions is highly relevant in the context of classical-verifier quantum protocols that make use of these functions (such as the ones from~\cite{mahadev2018classical, brakerski2018, gv19, mv20}). This was, in fact, recently used to develop depth-efficient \emph{proofs of quantumness}~\cite{liu2021depth}. 



With regards to the implications for quantum gravity research, we leave it open whether one can more formally describe our experiment by giving a Hamiltonian for the CFT whose grounds states encode the hard problem, instead of just positing that the CFT is prepared in one of those states. One would also need to understand how the geometry of the spacetime in the bulk is altered by evolution at the boundary and whether an observer in the bulk is able to measure any deformation. 
In fact, one objection that can be raised is that a constant entropy difference might not lead to a noticeable change in geometry. Assuming bulk observers can only differentiate between large-scale changes in geometry, this would translate into the question of whether \emph{entropy ratio} is intractable to estimate, rather than entropy difference. We show that this is indeed the case, in Appendix~\ref{sect:eratio}, by formally defining the entropy ratio problem and showing that it is equivalent (under Turing-Cook reductions) to the entropy difference problem. We also show how the \LWE\ construction can be directly extended to the entropy ratio case, showing a reduction from \LWE\ to estimating the entropy ratio of states produced by shallow circuits.



\subsection*{Acknowledgements} The authors would like to thank François Le Gall for pointing out the distinction between having bounded and unbounded fan-out gates in the quantum case. We are also grateful to Mehmet Burak Şahinoğlu, Thomas Vidick, Grant Salton, Hrant Gharibyan, Junyu Liu and Alexander Poremba for useful comments and discussions. Finally, we thank the anonymous referees for useful comments that improved our manuscript. 
AG acknowledges support from MURI Grant FA9550-18-1-0161 and the IQIM, an NSF Physics Frontiers Center (NSF Grant PHY-1125565) with support of the Gordon and Betty Moore Foundation (GBMF-12500028), as well as from Dr. Max R\"ossler, the Walter Haefner Foundation and the ETH Z\"urich Foundation. MJH acknowledges the FQXi large grant The Emergence of Agents from Causal Order.

\section{Preliminaries} \label{sect:prelim}

\subsection{Notation}
We write $\mathbb{N}$ for the set of natural numbers, $\mathbb{Z}$ for the set of integers, $\mathbb{Z}_q$ for the set of integers modulo $q$, $\mathbb{R}$ for the set of real numbers, $\mH$ for a finite-dimensional Hilbert space, using indices $\mH_A$, $\mH_B$ to specify distinct spaces. $\Lin(\mH)$ is the set of linear operators on $\mH$. We write $\Tr(\cdot)$ for the trace, and $\Tr_B:\Lin(\mH_A \otimes \mH_B )\to \Lin(\mH_A)$ for the partial trace. $\Pos(\mH)$ is the set of positive semidefinite operators and $\Density(\mH)=\{X\in \Pos(\mH):\Tr(X)=1\}$ the set of density matrices (also called states).  

Given $A, B \in \Lin(\mH)$, $\|A\|_1=\Tr\sqrt{A^\dagger A}$ is the Schatten $1$-norm, {$TD(A, B) = \frac{1}{2}\|A - B\|_1$} the trace distance, $F(A, B) = \left[ Tr \sqrt{\sqrt{A} B \sqrt{A}} \right]^2$ the fidelity, and $S(A) = -Tr(A \log A)$ the Von Neumann entropy.

Probability distributions will generally be over $n$-bit binary strings and so will be functions $\mathcal{D} : \{0, 1\}^n \to [0, 1]$ for which $\mathcal{D}(x) \geq 0$, for all $x \in \{0, 1\}^n$ and $\sum_{x \in \{0, 1\}^n} \mathcal{D}(x) = 1$. The Shannon entropy of a distribution is 
\begin{equation}
S(\mathcal{D}) = - \sum_{x \in \{0, 1\}^n} \mathcal{D}(x) \log(\mathcal{D}(x)).
\end{equation}
Here we are abusing notation since $S$ denotes both the Shannon entropy and the Von Neumann entropy and it will be clear from the context which notion of entropy is being used. We will also use $h : [0, 1] \to [0, 1]$ to denote the binary entropy function $h(\epsilon) = -\frac{1}{\epsilon} \log(\epsilon) - \frac{1}{1 - \epsilon} \log(1 - \epsilon)$. 

We write $\supp(\mathcal{D}) = \{x \; | \; \mathcal{D}(x) > 0 \}$ for the support of the distribution $\mathcal{D}$. The Hellinger distance between two distributions is:
\begin{equation}
H^2(\mathcal{D}_1, \mathcal{D}_2) = 1 - \sum_{x} \sqrt{\mathcal{D}_1(x) \mathcal{D}_2(x)}.
\end{equation}
Note that:
\begin{equation}
|| \mathcal{D}_1 - \mathcal{D}_2 ||_{TV} = \frac{1}{2} \sum_{x} | \mathcal{D}_1(x) - \mathcal{D}_2(x)| \leq H^2(\mathcal{D}_1, \mathcal{D}_2),
\end{equation}
where $|| \cdot ||_{TV}$ is the total variation distance.

For a finite set $S$, we write $x \leftarrow_U S$ to mean that $x$ is drawn uniformly at random from $S$. In general, $x \leftarrow_{\chi} S$ will mean that $x$ was drawn according to the distribution $\chi : S \to [0, 1]$ from $S$. We say that a function $\mu : \mathbb{N} \to \mathbb{R}$ is negligible if it goes to $0$ faster than any inverse-polynomial function, i.e. for any polynomial $p : \mathbb{N} \to \mathbb{R}$, $p(n) \mu(n) \to_{n \to \infty} 0$.

\subsection{Complexity theory} \label{sect:complexity}
In this paper we reference the standard complexity classes describing polynomial-time probabilistic and quantum computation, respectively, denoted \BPP\ and \BQP. 
In addition, we also utilize the complexity classes for computations that can be performed by polynomial-size circuits having polylogarithmic depth, \textsf{NC} for classical circuits and \textsf{QNC} for quantum circuits, logarithmic depth \NC\ and \QNC, as well as constant depth, \NCz\ and \QNCz, respectively. We highlight a common distinction between quantum and classical polylog-depth circuits. For classical \textsf{NC} circuits, gates have bounded fan-in but unbounded fan-out, for quantum \QNC\ circuits, gates have bounded fan-in and bounded fan-out. This distinction is important for considering constant-depth circuits, and \QNCf\ denotes the complexity class associated with constant-depth quantum circuits of unbounded fan-out (and bounded fan-in).
For formal definitions of these classes, as well as others mentioned in this paper, we refer to the Complexity Zoo~\cite{zoo}.
Here, we provide the definition of \QSZK, as it is the focus of our main results:

\begin{definition}[Quantum Statistical Zero-Knowledge (\QSZK)]
A promise problem $(L_{yes}, L_{no})$ belongs to \QSZK\ if there exists a uniform family of polynomial-size quantum circuits $V = \{V_n\}_{n > 0}$, known as a verifier such that the following conditions are satisfied:
\begin{itemize}
\item Completeness. For each $x \in L_{yes}$, there exists a prover $P(x)$ exchanging  polynomially-many quantum messages with $V(x)$ that makes $V(x)$ accept with probability greater than $2/3$,
\item Soundness. For each $x \in L_{no}$, any prover $P(x)$ exchanging polynomially-many quantum messages with $V(x)$ will make $V(x)$ accept with probability at most $1/3$,
\item Zero-Knowledge. There exists a uniform family of polynomial-size quantum circuits $S = \{ S_n \}_{n > 0}$, known as a simulator, as well as a negligible function $\varepsilon : \mathbb{N} \rightarrow [0,1]$ such that for all $x \in L_{yes}$, $S(x)$ produces a state that is $\varepsilon$-close in trace distance to the transcript of $V(x) \leftrightarrow P(x)$.
\end{itemize}
\end{definition}

As was shown by Ben-Aroya, Schwartz and Ta-Shma~\cite{bst}, the following problem is \QSZK-complete under polynomial-time reductions:
\begin{definition}[Quantum Entropy Difference (\QED)~\cite{bst}]
Let $C_1$ and $C_2$ be quantum circuits acting on $n + k$ qubits. Define the following $n$-qubit mixed states:
\begin{equation}
\rho_1 = Tr_{k} (C_1 \ket{00...0} \bra{00...0} C_1^{\dagger}) \quad \quad \quad \quad
\rho_2 = Tr_{k} (C_2 \ket{00...0} \bra{00...0} C_2^{\dagger})
\end{equation}
Given $n$, $k$ and descriptions of $C_1$ and $C_2$ as input, decide whether:
\begin{equation}
S(\rho_1) \geq S(\rho_2) + 1/poly(n+k)
\end{equation}
or
\begin{equation}
S(\rho_2) \geq S(\rho_1) + 1/poly(n+k)
\end{equation}
promised that one of these is the case.
\end{definition}
Note that the entropy difference can always be made constant by considering \emph{parallel}\footnote{Hence the new circuits have the same depth as the original ones.} repeated version of $C_1$ and $C_2$. Specifically, if the entropy difference for $C_1$ and $C_2$ is $D = 1/poly(n+k)$, then the difference for $C_1^{\otimes D}$ and $C_2^{\otimes D}$, acting on $(n+k)D$ qubits and when tracing out $kD$ qubits, will be $O(1)$.
For this reason, throughout this paper we will consider the entropy difference problem (and all its variants) with a constant entropy difference. 

The classical/probabilistic analogues of the above definitions are represented by the class \SZK\ and the Entropy Difference (\ED) problem. In analogy to \QED, \ED\ is defined as the problem of distinguishing the output entropies of classical circuits having uniformly random inputs. It was shown in~\cite{goldreich} that \ED\ is \SZK-complete.

\subsection{Cryptography} \label{subsect:lwe}
In this section we define the relevant cryptographic notions that we will be using.

\subsubsection{Learning with Errors (LWE)}
The Learning with Errors (\LWE) problem, first introduced by Regev~\cite{lwe}, is the problem of solving an approximate system of linear equations. In other words, given a matrix $A \in \mZ_q^{n \times m}$, as well as a vector $y = As + e \; (mod \; q)$, with $y \in \mZ_q^m$, $s \in \mZ_q^n$ and $e \leftarrow_{\chi^m} \mZ_q^m$ (here $\chi$ denotes a truncated Gaussian distribution that is described below), the problem is to determine $s$. This is known as the \emph{search} version of \LWE\ though there is also a \emph{decision} version. The decision version is to distinguish between a vector of the form $As + e$ and a uniformly random vector $u \leftarrow_U \mathbb{Z}^m_q$. We refer to the tuple $(A, As + e)$ as an \emph{LWE sample}.

The truncated Gaussian distribution is defined as follows. For a positive real $B$ and positive integer $q$, the truncated discrete Gaussian distribution over $\mZ_q$ with parameter $B$ is supported on $\{x\in\mZ_q:\,\|x\|\leq B\}$ and has density
\begin{equation}\label{eq:d-bounded-def}
 D_{\mZ_q,B}(x) \,=\, \frac{e^{\frac{-\pi\lVert x\rVert^2}{B^2}}}{\sum\limits_{x\in\mZ_q,\, \|x\|\leq B}e^{\frac{-\pi\lVert x\rVert^2}{B^2}}} \;.
\end{equation}
For a positive integer $m$, the truncated discrete Gaussian distribution over $\mZ_q^m$ with parameter $B$ is supported on $\{x\in\mZ_q^m:\,\|x\|\leq B\sqrt{m}\}$ and has density
\begin{equation}\label{eq:d-bounded-def-m}
\forall x = (x_1,\ldots,x_m) \in \mZ_q^m\;,\qquad D_{\mZ_q^m,B}(x) \,=\, D_{\mZ_q,B}(x_1)\cdots D_{\mZ_q,B}(x_m)\;.
\end{equation}

As shown in \cite{lwe,PRS17}, for any $\alpha>0$ such that $\sigma = \alpha q \ge 2 \sqrt{n}$ the LWE problem, with error vector $e$ sampled from $D_{\mZ_q,\sigma}$, is at least as hard (under a quantum poly-time reduction) as approximating the shortest independent vector problem (SIVP) to within a factor of $\gamma = \tilde{O}({n}/\alpha)$ in \emph{worst case} dimension $n$ lattices.

Typically, an instance of \LWE\ is parametrized by a number $\lambda > 0$, referred to as the \emph{security parameter}. We think of all the LWE parameters $n, m, q$ as being functions of $\lambda$. The so-called (non-uniform) ``\LWE\ assumption'' is that no quantum polynomial-time procedure can solve LWE (either in the search or decision version) with more than a negligible advantage in $\lambda$, even when given access to a quantum polynomial-size state that depends on $\lambda$.

\subsubsection{Claw-free functions}

For the reduction from \lwe\ to \QEDlog\ we make use of ``extended trapdoor claw-free functions (ETCF),'' introduced in~\cite{brakerski2018, mahadev2018classical}.
A claw-free function is a type of hash function which has pairs of collisions known as ``claws''. The function should be efficient to evaluate but it should be intractable to find claws (given a description of the function). This is referred to as the ``claw-free'' property and is based on some computational assumption. In~\cite{brakerski2018, mahadev2018classical}, the authors construct such functions for which the claw-free property is based on the intractability of LWE. The trapdoor property means that the functions can be inverted efficiently, given a secret trapdoor, however we do not require this property for our results. Finally, ETCFs are pairs consisting of claw-free functions and injective functions. The latter do not have collisions, as a result of being injective. However, the injective functions are constructed so that they are indistinguishable from the claw-free functions. This property is known as \emph{injective invariance} and is the key property we will use to prove the hardness of entropy estimation. We now give a simplified definition of ETCFs and refer the reader to~\cite{brakerski2018, mahadev2018classical} for the full version.

\begin{definition}[Extended Trapdoor Claw-free Functions (ETCFs)~\cite{brakerski2018,mahadev2018classical}] \label{def:etcfs}
  A pair of functions $f,g : \mathcal{X} \to \mathcal{Y}$ forms an Extended Trapdoor Claw-free Function (ETCF) pair if they satisfy the following properties:
  \begin{itemize}
    \item \textbf{Efficient generation}. There exists a polynomial-time algorithm $\textsf{Gen}\left(1^{\lambda}\right)$, $\lambda > 0$, that generates circuit descriptions for $f$ and $g$. Here and throughout the paper, $f$ will denote the injective function and $g$ will denote the claw-free function. In other words, $f$ is a 1-to-1 mapping, while $g$ is a 2-to-1 mapping.
    \item \textbf{Efficient evaluation}. There exist polynomial-time evaluation algorithms for $f$ and $g$. In other words, given any $x \in \mathcal{X}$, both $f(x)$ and $g(x)$ can be computed in polynomial time.  
    \item \textbf{Claw-free}. There is no polynomial-time algorithm that, given a description of $g$, is able to output $x_0, x_1 \in \mathcal{X}$ such that $g(x_0) = g(x_1)$, with probability greater than $\negl(\lambda)$. 
    \item \textbf{Trapdoor}. The generation algorithm, $\textsf{Gen}\left(1^{\lambda}\right)$, also outputs strings $t_f$ and $t_g$ (known as trapdoors) which allow for efficiently inverting $f$ and $g$, respectively. In other words, there exists a polynomial-time algorithm that given $t_f$ ($t_g$) and $f(x)$ ($g(x)$) outputs $x$ ($x_0$ and $x_1$ such that $g(x_0)=g(x_1)=g(x)$), for all $x \in \mathcal{X}$.
    \item \textbf{Injective invariance} For any polynomial-time algorithm $\mathcal{A}$, which receives as input a description of either $f$ or $g$, it is the case that:
      \[ \left| \Pr \left[ 0 \leftarrow \mathcal{A}(f) \right] - \Pr \left[ 0 \leftarrow \mathcal{A}(g) \right] \right| = \negl(\lambda). \]
    In other words, no efficient algorithm can distinguish between the claw-free function and the injective function, given a description of either.
  \end{itemize}
\end{definition}

\subsubsection{Randomized encodings}
To prove the reduction from \LWE\ to the constant depth version of the entropy difference problem, we make use of randomized encodings. The following definitions and results are taken from~\cite{aik}:

\begin{definition}[Uniform randomized encoding~\cite{aik}] \label{def:randenc}
Let $f : \{0, 1\}^* \to \{0, 1\}^*$ be a polynomial-time computable function and $l(n)$ an output length function such that $|f(x)| = l(|x|)$, for every $x \in \{0, 1\}^*$. We say that a function $\hat{f} : \{0, 1\}^* \times \{0, 1\}^* \to \{0, 1\}^*$ is a $\delta(n)$-correct, $\varepsilon(n)$-private randomized encoding of $f$, if it satisfies the following properties:
\begin{itemize}
\item \textbf{Length regularity.} There exist polynomially-bounded and efficiently computable length functions $m(n)$, $s(n)$ such that for every $x \in \{0, 1\}^{n}$ and $r \in \{0, 1\}^{m(n)}$ we have that $|\hat{f}(x, r)| = s(n)$.
\item \textbf{Efficient evaluation.} There exists a polynomial-time evaluation algorithm that, given $x \in \{0, 1\}^*$ and $r \in \{0, 1\}^{m(n)}$, outputs $\hat{f}(x, r)$.
\item $\delta$-\textbf{correctness.} There exists a polynomial-time algorithm $C$, called a decoder, such that for every input $x \in \{0, 1\}^n$, 
\begin{equation}
\Pr_{r \leftarrow_U \{0, 1\}^{m(n)}}[C(1^n, \hat{f}(x, r)) \neq f(x)] \leq \delta(n).
\end{equation}
\item $\varepsilon$-\textbf{privacy.} There exists a probabilistic polynomial-time algorithm $S$, called a simulator such that for every $x \in \{0, 1\}^n$ and $r \leftarrow_U \{0, 1\}^{m(n)}$,
\begin{equation}
|| S(1^n, f(x)) - \hat{f}(x, r) ||_{TV} \leq \varepsilon(n).
\end{equation}
\end{itemize}
\end{definition}
Correctness (and privacy, respectively) can be \emph{perfect} when $\delta = 0$ ($\varepsilon = 0$, respectively), or \emph{statistical} when $\delta(n)$ ($\varepsilon(n)$, respectively) is a negligible function in $n$. Thus, a \emph{perfect randomized encoding} is one that has perfect correctness and perfect privacy\footnote{In fact a perfect randomized encoding needs to satisfy two further properties called \emph{balance} and \emph{stretch-preservation}~\cite{aik}, though we omit them here since they are not important for our results.}. Similarly, a \emph{statistical randomized encoding} is one that has statistical correctness and statistical privacy.

\begin{lemma}[Unique randomness~\cite{aik}] \label{lemma:uniquerand}
Suppose $\hat{f}$ is a perfect randomized encoding of $f$. Then,
\begin{itemize}
\item[(a)] for any input $x$, $\hat{f}(x, \cdot)$ is injective, i.e. there are no distinct $r, r'$, such that $\hat{f}(x, r) = \hat{f}(x, r')$, 
\item[(b)] if $f$ is a permutation then so is $\hat{f}$.
\end{itemize}
\end{lemma}

\begin{theorem}[\cite{aik}] \label{thm:pren}
Any function $f \in \NC$ admits a perfect randomized encoding in $\mathsf{NC}^0_4$. Moreover, constructing the randomized encoding of $f$ can be achieved in time polynomial in the size of the circuit that evaluates $f$.
\end{theorem}
\noindent Here $\mathsf{NC}^0_4$ denotes the set of uniform constant-depth circuits having output locality 4. In other words, each output bit depends on at most 4 input bits.

It should be noted here that the randomized encoding constructions from~\cite{aik} use gates of unbounded fan-out (at least in the first circuit layer).

\section{Entropy difference for shallow circuits with unbounded fan-out gates}
In this section we examine the hardness of the entropy difference problem for circuits whose depth scales at most (poly)logarithmically with the size of the input. Throughout this section, unless otherwise stated, we assume that gates in the circuits have bounded fan-in and unbounded fan-out. We will address the case of bounded fan-out at the end of this section.

We start by formally defining the entropy difference problem, in both the classical and the quantum cases, for circuits with depth scaling as $\delta(n)$, for inputs of size $n$ and where $\delta : \mathbb{N} \to \mathbb{N}$, is a monotonically increasing function:

\begin{definition}[\EDf]
Let $C_1$ and $C_2$ be reversible boolean circuits acting on $n + k$ bits such that $depth(C_1) \leq \delta(n+k)$, $depth(C_2) \leq \delta(n+k)$. Let
\begin{equation}
\mathcal{D}_1(y) = \Pr_{x \leftarrow_{U} \{0, 1\}^{n+k}} \left[ y = Tr_{k}(C_1(x)) \right]  \quad \quad \quad \mathcal{D}_2(y) = \Pr_{x \leftarrow_{U} \{0, 1\}^{n+k}} \left[ y = Tr_{k}(C_2(x)) \right]
\end{equation}
denote the output distributions of the circuits when restricted to the first $n$ output bits and with the input chosen uniformly at random.

Given $n$, $k$ and descriptions of $C_1$ and $C_2$ as input, decide whether:
\begin{equation}
S(\mathcal{D}_1) \geq S(\mathcal{D}_2) + 1 \quad\quad \text{or} \quad\quad
S(\mathcal{D}_2) \geq S(\mathcal{D}_1) + 1
\end{equation}
promised that one of these is the case\footnote{As mentioned in Subsection~\ref{sect:complexity}, the gap in entropy difference can be $1/poly(n + k)$, but this can always be made constant by simply taking parallel repeated versions of the circuits and using the fact that entropy is additive. We restrict to the case where the entropy difference is at least $1$, unless otherwise specified.}.
\end{definition}

\begin{definition}[\QEDf]
Let $C_1$ and $C_2$ be quantum circuits acting on $n + k$ qubits such that $depth(C_1) \leq \delta(n+k)$, $depth(C_2) \leq \delta(n+k)$. Define the following $n$-qubit mixed states:
\begin{equation}
\rho_1 = Tr_{k} (C_1 \ket{00...0} \bra{00...0} C_1^{\dagger}) \quad \quad \quad \quad
\rho_2 = Tr_{k} (C_2 \ket{00...0} \bra{00...0} C_2^{\dagger})
\end{equation}
Given $n$, $k$ and descriptions of $C_1$ and $C_2$ as input, decide whether:
\begin{equation}
S(\rho_1) \geq S(\rho_2) + 1 \quad\quad \text{or} \quad\quad
S(\rho_2) \geq S(\rho_1) + 1
\end{equation}
promised that one of these is the case.
\end{definition}

Note that if $\delta(n) = O(\log(n))$ we obtain the definitions for \EDlog\ and \QEDlog, respectively, and if $\delta(n) = O(1)$, we obtain the definitions for \EDconst\ and \QEDconst, respectively\footnote{We slightly abuse notation here since $O(\log(n))$ and $O(1)$ are sets of functions and so, correspondingly, \EDlog\ (\QEDlog) and \EDconst\ (\QEDconst) will also be sets of problems. However, our results are valid for any instances of the problems in these classes.}. For the case when the depth of the circuits is constant (or sub-logarithmic, in general), we need to distinguish between cases when the gates have bounded or unbounded fan-out (fan-in will always be assumed to be bounded). As such, \EDconst\ and \QEDconst\ will be problem instances in which the circuits have gates of unbounded fan-out, whereas we denote as $ED_{O(1)}^b$ and $QED_{O(1)}^b$ the analogs of \EDconst\ and \QEDconst\ with gates of bounded fan-out. Note that when $\delta(n) = poly(n)$, we recover the original definitions of \ED\ and \QED, respectively. As those problems are complete for \SZK\ and \QSZK, we proceed to define the analogous classes \SZKf\ and \QSZKf\ of problems that poly-time reduce to \EDf\ and \QEDf.

\begin{definition}[\SZKf, \QSZKf]
We let \SZKf\ (\QSZKf) consist of all promise problems $(L_{yes}, L_{no})$ for which there exists a deterministic polynomial-time reduction to \EDf\ (\QEDf).
\end{definition}

\noindent Taking $\delta(n) = polylog(n)$, we now show that these classes have the following characterisation:

\begin{lemma}
\SZKlogd\ (\QSZKlogd) is the set of promise problems that admit a (quantum) statistical zero-knowledge protocol in which the verifier circuit and the simulator circuit have depth $polylog(n)$, for inputs of size $n$.
\end{lemma}
\begin{proof}
To prove this result we need to show that \EDlogd\ (\QEDlogd) is contained and complete for the class of problems that admit a (quantum) statistical zero-knowledge protocol in which the verifier circuit and the simulator circuit have depth $polylog(n)$.
This follows immediately from the proofs that \ED\ (\QED) is contained and complete for \SZK\ (\QSZK), from~\cite{vidick2016quantum}, by replacing all circuits in those proofs with the corresponding circuits of depth $polylog(n)$.
\end{proof}

\subsection{Entropy difference is easier for shallow circuits}
With the above definitions, let us now address the question of whether \EDlogd\ (\QEDlogd) is \SZK-complete (\QSZK-complete). From the above lemma, we see that this is the same as asking whether $\SZKlogd = SZK$ ($\QSZKlogd = QSZK$). In other words, can any \SZK\ (\QSZK) protocol be turned into a protocol having polylog-depth circuits for the simulator and verifier? Providing an \emph{unconditional} negative answer seems difficult without proving explicit circuit lower-bounds. 
Instead, we will give ``complexity-theoretic evidence'' that the answer is no. We first show that there exists an oracle, $\mathcal{O}$, such that, relative to $\mathcal{O}$, \SZKlogd\ $\neq$ \SZK\ (\QSZKlogd $\neq$ \QSZK).

\begin{theorem}
There exists an oracle $\mathcal{O}$ such that $\mathsf{SZK^{\mathcal{O}}_{polylog}} \neq \mathsf{SZK}^{\mathcal{O}}$ and $\mathsf{QSZK^{\mathcal{O}}_{polylog}} \neq \mathsf{QSZK}^{\mathcal{O}}$.
\end{theorem}
\begin{proof}
The proof will be primarily for the quantum case, since the classical case is analogous, though we will specify whenever there is a distinction between the two.
The oracle in our proof will be identical to the one of Chia, Chung and Lai~\cite{chiachunglai}, showing the separation $\mathsf{BPP}^{\mathsf{QNC}^{\mathcal{O}}} \neq \mathsf{BQP}^{\mathcal{O}}$. 
Their oracle provides access to a function $f : \{0, 1\}^n \to \{0, 1\}^n$ that is promised to be either $1$-to-$1$ or $2$-to-$1$. However, the oracle does not give direct query access to $f$. Instead, the oracle allows for the querying of $d + 1$ functions $f_0, f_1 ... , f_d : \{0, 1\}^n \to \{0, 1\}^n$ such that $f = f_d \circ ... \circ f_0$, for some $d > 0$. The functions $f_0, f_1, ... , f_{d-1}$ are $1$-to-$1$ functions and $f_{d}$ is either a $1$-to-$1$ function or a $2$-to-$1$ function depending on whether $f$ is $1$-to-$1$ or $2$-to-$1$.
This is referred to as a \emph{$d$-shuffling oracle}. Its purpose is to force any algorithm that attempts to query $f$ to first evaluate the $d+1$ functions. This is achieved by having the image of each function be a random subset of its co-domain. In other words, each function will be defined as $f_i : S_i^{(i)} \to S_{i+1}^{(i+1)}$, with $S_i^{(i)} \subseteq \{0, 1\}^n$. The input domain, however, will be a set $S_0 \subseteq \{0, 1\}^n$ chosen uniformly at random from subsets of $n$-bit strings. Thus, the image of $f_1$ on $S_0$ will be $Im_{S_0}(f_1) = f_1(S_0) = S_1$ and in general $S_i = Im_{S_{i-1}}(f_i)$.

The problem that Chia, Chung and Lai define relative to this oracle is to determine whether the function $f : S_0 \to S_{d+1}$ is $1$-to-$1$ or $2$-to-$1$. In the latter case, the function also has Simon's property so that the problem (called \emph{$d$-shuffling Simon's problem}, or $d$-\textsf{SSP}) can be solved efficiently in quantum polynomial time, thus showing containment in $\BQP^{\mathcal{O}}$. Using the properties of the $d$-shuffling oracle it is possible to show that no quantum circuit of depth smaller than $d$ can solve the problem, even when alternating these quantum circuits with classical circuits of polynomial depth. Thus, taking $d = \Omega(n)$ is sufficient to show that the problem is not contained ${\BPP^{\mathsf{QNC}}}^{\mathcal{O}}$.

For the proof of our result we also consider the $d$-\textsf{SSP} problem. Since we already know that the problem is in $\BQP^{\mathcal{O}}$ this immediately implies that the problem is also in $\QSZK^{\mathcal{O}}$. For the classical case, we would also need to show that the problem is contained in $\SZK^{\mathcal{O}}$. This follows from the fact that \ED\ is in \SZK\ and the problem of determining whether a function is $1$-to-$1$ or $2$-to-$1$ reduces to \ED\footnote{This is because $f$ evaluated on a uniformly random $n$-bit string will have maximum entropy $n$, when $f$ is $1$-to-$1$ and entropy $n-1$ when the function is $2$-to-$1$.}. 
We therefore need to show that the problem is not contained in $\SZKlogd^{\mathcal{O}}$ and $\QSZKlogd^{\mathcal{O}}$.

\begin{figure}[ht!]
\begin{center}
  \includegraphics[width=0.7\linewidth]{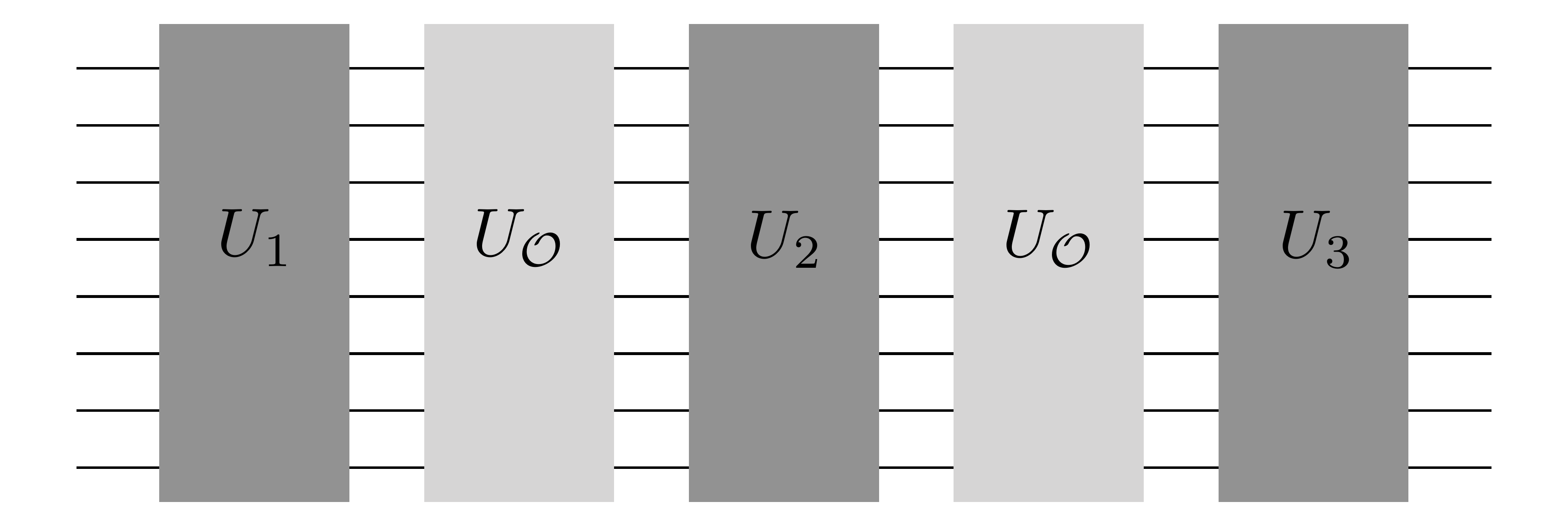}
\caption{Circuit with oracle calls. The unitaries $U_j$ are assumed to be polylogarithmic depth quantum circuits. The unitaries $U_{\mathcal{O}}$ represent calls to the oracle $\mathcal{O}$. The circuit has polylogarithmic depth overall.}
\label{fig:simcircuit}
\end{center}
\end{figure}

To do this, consider a \QSZKlogd\ protocol in which the verifier, the prover and the simulator all have access to the oracle $\mathcal{O}$. In such a protocol, the verifier and the simulator are circuits of depth $O(polylog(n))$ that generically consist of alternating sequences of polylog-depth circuits and calls to the oracle, as shown in Figure~\ref{fig:simcircuit}. For a fixed input state that starts off uncorrelated with the oracle, let us examine the output state of such a circuit in the cases when the function is injective and when it is a $2$-to-$1$ function. 
We will denote the oracle as $\mathcal{O}(f)$ in the former case and as $\mathcal{O}(g)$ in the latter\footnote{The interpretation of this notation is that the $d$-shuffling oracle is providing ``shuffled'' access to either $f$ or $g$.}. We also denote the circuit under consideration, when given access to the oracle, as $\COf$ and $\COg$, respectively.
These can be expressed as follows
\begin{equation}
\COf = U_m U_{\mathcal{O}(f)} U_{m-1} ... U_{\mathcal{O}(f)} U_{1} \quad \quad \COg = U_m U_{\mathcal{O}(g)} U_{m-1} ... U_{\mathcal{O}(g)} U_{1}
\end{equation}
with $m = polylog(n)$ and where each $U_i$ is a circuit of depth one.
We denote the input state to the circuit as\footnote{The analysis remains unchanged if the input state is a mixed state.} $\ket{\psi(0)}$. We will also write $\ket{\psi^f(i)} = U_{\mathcal{O}(f)} U_i \ket{\psi^f(i-1)}$ and $\ket{\psi^g(i)} =U_{\mathcal{O}(g)} U_i \ket{\psi^g(i-1)}$. 
Note that
\begin{equation}
\ket{\psi^f(m)} = \COf \ket{\psi(0)} \quad \quad \ket{\psi^g(m)} = \COg \ket{\psi(0)} 
\end{equation}
Following a similar analysis to that of~\cite{chiachunglai}, we have that
\begin{multline}
TD(\ket{\psi^f(m)}, \ket{\psi^g(m)}) \leq 
TD(\ket{\psi^f(m)}, U_{\mathcal{O}(f)} U_{m} \ket{\psi^g(m - 1)}) + \\ TD(U_{\mathcal{O}(f)} U_{m} \ket{\psi^g(m - 1)}, \ket{\psi^g(m)})
\end{multline}
which can be extended to
\begin{equation}
TD(\ket{\psi^f(m)}, \ket{\psi^g(m)}) \leq \sum_{i=1}^{m} TD(U_{\mathcal{O}(f)} U_{i} \ket{\psi^g(i - 1)}, \ket{\psi^g(i)})
\end{equation}
Both of these equations follow from the triangle inequality. Using a hybrid argument, we have that
\begin{equation}
TD(\ket{\psi^f(m)}, \ket{\psi^g(m)}) \leq m \max_{i \leq m} \; TD(U_{\mathcal{O}(f)} U_{i} \ket{\psi^g(i - 1)}, U_{\mathcal{O}(g)} U_{i} \ket{\psi^g(i - 1)})
\end{equation}
Finally, as in~\cite[Theorem 6.1]{chiachunglai}, one can use the one-way to hiding lemma~\cite{ambainis2019quantum} to show that for $i \leq d$,
\begin{equation}
TD(U_{\mathcal{O}(f)} U_{i} \ket{\psi^g(i - 1)}, U_{\mathcal{O}(g)} U_{i} \ket{\psi^g(i - 1)}) \leq \frac{poly(n)}{2^n}
\end{equation}
and since $m = polylog(n)$, $d = \Omega(n)$, for sufficiently large $n$, $m < d$, hence
\begin{equation}
TD(\ket{\psi^f(m)}, \ket{\psi^g(m)}) \leq m \frac{poly(n)}{2^n}.
\end{equation}
If we now consider $C^{\mathcal{O}}$ to be the simulator circuit in a \QSZKlogd\ (\SZKlogd) protocol, we see that the simulator produces nearly identical transcripts irrespective of whether the oracle function is $1$-to-$1$ or $2$-to-$1$.
This means that, for the ``yes'' instances (the function being $1$-to-$1$), the transcript that the verifier circuit acts on, upon its interaction with the prover, is \emph{almost completely uncorrelated} with the oracle itself. Stated differently, the transcript is $poly(n)/2^n$-close in trace distance to a transcript for a ``no'' instance. 
Thus, the interaction with the prover can provide the verifier with at most a $poly(n)/2^n$ advantage in deciding the problem correctly. Since the verifier circuit itself is polylogarithmic in depth, from the above analysis (and the result of~\cite{chiachunglai}), it follows that if the oracle function type is equiprobable to be $1$-to-$1$ or $2$-to-$1$, the resulting \QSZKlogd\ (\SZKlogd) protocol will decide correctly with probability at most $1/2 + poly(n)/2^n$.
This concludes the proof.
\end{proof}

It should be noted that, following~\cite{chiachunglai}, the above result extends to circuits of depth strictly less than $d$. The key insight of the proof is the fact that the shuffling oracle requires circuits of depth at least $d$ in order to obtain an output that is non-negligibly correlated with the oracle type. This also results in the more fine-grained result $(Q)ED^{\mathcal{O}}_{\delta(n)} \leq_P (Q)ED^{\mathcal{O}}_{2\delta(n) + 1}$. The $2\delta(n) + 1$ upper bound comes from the fact that an instance of $d$-$\mathsf{SSP}$ can be solved with depth $2d + 1$.
For our case, since $d = \Omega(n)$, if we were to look strictly at instances of \EDlogd\ and \QEDlogd\, in which the circuits under consideration can query the oracle, the number of queries is too small for there to be any noticeable difference in the output entropies. Thus, relative to the shuffling oracle, \EDlogd\ and \QEDlogd\ are strictly weaker than \ED\ and \QED.

Let us now consider a different argument for why it is unlikely that entropy difference with shallow circuits is as hard as with general poly-size circuits. We will focus specifically on the quantum case for circuits of logarithmic depth and show the following:

\begin{theorem} \label{thm6}
If there exists a polynomial-time reduction from a \QSZK\ protocol with a log-depth verifier to a \QSZKlog\ protocol which preserves the transcript of the \QSZK\ protocol, then $\BQP=\BPP^{\QNC}$.
\end{theorem}
\begin{proof}
It is straightforward to show that $\BQP\subseteq\QSZK$: the quantum verifier ignores the prover and can decide any language in \BQP\ . However, one can also give a \QSZK\ protocol for any language in \BQP\ where there is non-trivial interaction between a prover and verifier. Furthermore, we show that such a protocol only requires the verifier's circuit to be log depth. To show this, we adapt the proof by Rosgen that \QIP\ only requires log depth quantum verifiers \cite{rosgenQIP}.

Given a polynomial-time quantum circuit on $n$ qubits of the form $C = U_m U_{m-1} ... U_1$, the following state $\ket{\psi_U}$ can be constructed by another \BQP\ circuit,

\begin{equation}
\ket{\psi_U} = \ket{0...0} \otimes U_1 \ket{0...0} \otimes ... \otimes U_m U_{m-1} ... U_1 \ket{0...0}.
\end{equation}
Notice the similarity to a Feynman-Kitaev history state except where one takes the tensor product of unitaries applied to the input, and not the superposition. The state $\ket{\psi_U}$ stores the final state of the circuit above in the rightmost $n$ qubits. Therefore, any language in \BQP\ can be decided by measuring the relevant qubit in these rightmost $n$ qubits in $\ket{\psi_U}$.

Now we can have a \QSZK\ protocol where an honest prover gives the verifier the state $\ket{\psi_U}$, and thus a verifier has the ability to decide any language in \BQP\ when given this state. In the case of a dishonest prover, we can use the techniques described by Rosgen, to verify that the state given by the prover is $\ket{\psi_U}$. For convenience of explanation, we divide the state given by the prover up into $m$ ``registers" of $n$ qubits, where the first register should be in the state $\ket{0...0}$, the second register in the state $U_1\ket{0...0}$, and so on. The idea for verifying that the state is $\ket{\psi_U}$ is to pairwise compare the $j$th and $(j+1)$th registers with SWAP tests. More precisely the verification circuit will apply $U_{j+1}$ to the $j$th register and perform a SWAP test on the $j$th and $(j+1)$th registers: if the states are the same then swapping leaves the states invariant and the circuit accepts, otherwise it rejects. Therefore, after $n$ SWAP tests of this form, if all tests accept, then with high probability the state is $\ket{\psi_U}$. Importantly, all of the SWAP tests to compare registers can be done in log depth \cite{rosgenQIP}, thus the verifier's circuit is a log depth quantum circuit. 

In the above protocol we have outlined how the verifier can verify  
the state $\ket{\psi_U}$, but we have not shown that it satisfies the property of statistical zero-knowledge. Note that the state $\ket{\psi_U}$ produced by the prover (such that an input is accepted) can be generated by a polynomial time quantum circuit. Therefore, in the case of a \QSZK\ protocol, the simulator could produce this state $\ket{\psi_U}$, and since this state is the whole transcript of the protocol, the protocol has the property of zero-knowledge. 

If we assume that there exists a polynomial-time reduction from a \QSZK\ protocol to a \QSZKlog\ protocol which preserves the transcript of the \QSZK\ protocol, then the above \QSZK\ protocol for deciding \BQP\ can be turned into a \QSZKlog\ protocol. Therefore, a simulator $S$ must be able to produce a state very close to $\ket{\psi_U}$ with a log-depth quantum circuit, in the case that the input $x$ is in the language $L_{yes}$. Furthermore, since \BQP\ is closed under complement, we can take the complement of the language above, and have a \QSZKlog\ protocol for this, and thus another simulator $S'$ that produces a state close to $\ket{\psi_U}$ for $x\in L_{no}$.

Now we have the situation where if the conditions of the theorem hold, we have two log-depth quantum circuits corresponding to the simulators $S$ and $S'$ above that can be used to decide membership of a language in \BQP: if $x\in L_{yes}$ then $S$ will produce the state $\ket{\psi_U}$, but $S'$ could produce anything; if $x\in L_{no}$ then $S'$ produces the correct state $\ket{\psi_U}$. To decide which is which, a verifier can apply the SWAP tests outlined above individually on both of the states generate by $S$ and $S'$: at least one of the two states will satisfy the tests and correctly accept or reject. We can now leverage these observations to prove the theorem. 

To collect the observations together, we have pointed out that if the conditions of the theorem hold, there are log-depth quantum circuits $S$ and $S'$ that generate states $\ket{\psi_U}$, which can be used by a log-depth quantum circuit to decide membership in \BQP. Thus we could decide any language in \BQP\ with a $\BPP^{\QNC}$ algorithm in the following way: a \BPP\ machine computes the reduction from a \QSZK\ protocol to a \QSZKlog\ protocol, feeds the circuit descriptions of the log-depth simulators $S$ and $S'$ to a \QNC\ oracle, which can then produce states of the form $\ket{\psi_U}$, and carry out the necessary SWAP tests to verify this state. If the SWAP tests are passed for at least one of the two states produced by $S$ and $S'$, then the accept/reject measurement is performed on it (again by the oracle), and the algorithm accepts if the oracle accepts, or rejects otherwise.
\end{proof}

It is worthwhile pointing out that this trick of deciding languages in \BQP\ with a $\BPP^{\QNC}$ algorithm does not obviously generalise to other complexity classes. First, we would need that there are quantum statistical zero-knowledge protocols for the class, and use the property of closure under complementation. Naturally, \QSZK\ satisfies both of these properties, but an arbitrary protocol for languages in \QSZK\ has multiple rounds of communication between prover and verifier. Our construction uses the fact that a \QSZK\ protocol for any language in \BQP\ has a single round of communication from prover to verifier, which facilitates verification of a state in log depth. It is far from obvious how to construct such a verification procedure for an arbitrary language in \QSZK. 


\subsection{Hardness based on Learning with Errors}
In the previous subsection we showed that the polylogarithmic-depth version of the entropy difference problem, in both the classical and quantum case, is unlikely to be as hard as the polynomial-depth variant. This raises the question of whether the problem becomes tractable for polynomial-time classical or quantum algorithms. Here we give indication that the answer is no by proving a reduction from \LWE\ to \EDlog. 
Using techniques from~\cite{aik}, we then strengthen this result by also showing a reduction from \LWE\ to \EDconst. We begin with the log-depth case:

\begin{theorem} \label{thm:lweedlog}
\lwe\ $\leq_P$ \EDlog.
\end{theorem}
\begin{proof}
The proof uses the ETCF functions defined in Subsection~\ref{subsect:lwe}. As mentioned, an ETCF family consists of an injective function and a $2$-to-$1$ (or claw-free) function and there exists a reduction from \lwe\ to the problem of distinguishing the two functions given their descriptions\footnote{This is the injective invariance property of Definition~\ref{def:etcfs}.}. By showing that the two functions can be evaluated using circuits of logarithmic depth, we will extend this reduction to \EDlog.
While it has already been shown that certain cryptographic functions based on \LWE\ can be performed in \NC~\cite{banerjee2012pseudorandom}, our result requires that we show this for the circuits that we construct from an ETCF function family.

The ETCF functions we consider will be the same as the ones from~\cite{mahadev2018classical, brakerski2018}:
\begin{equation} \label{eqn:etcf}
f(b, x) = Ax + b \cdot u + e \; (mod \; q) \quad \quad \quad g(b, x) = Ax + b \cdot (As + e') + e \; (mod \; q)
\end{equation}
with\footnote{It should be noted that $q$ itself is a function of $n$. In general $q$ is taken to be suprapolynomial in $n$, though in recent constructions $q$ can also be polynomial in $n$ while still preserving the hardness of the \LWE\ instance~\cite{brakerski2013classical}.} $q \geq 2$, $b \in \{0, 1\}$, $x \in \mZ_q^n$, $s \in \mZ_q^n$, $A \in \mZ^{n \times m}_q$, $u, e' \in \mZ_q^m$, $e \leftarrow_{ D_{\mZ_q,B}^m} \mZ^m$.
These functions will be ETCF even when $A$, $e'$ and $u$ are chosen at random as follows: $A \leftarrow_U \mZ_q^{n \times m}$, $u \leftarrow_U \mZ_q^m$, $e' \leftarrow_{ D_{\mZ_q,B'}^m} \mZ^m$.
The values $B'$ and $B$ are the ones from~\cite{mahadev2018classical, brakerski2018} (there denoted as $B_P$ and $B_V$), and determine where the Gaussian distributions are truncated. Specifically, $B' = \frac{q}{C_T \sqrt{mn \log(q)}}$, where $C_T$ is a fixed constant and $B$ is chosen so that $B'/B$ is super-polynomial in $n$. 
Since it was already shown in~\cite{mahadev2018classical} that the above functions are ETCF, we need only show that they can be evaluated by log-depth circuits.

First of all note that the functions from Equation~\ref{eqn:etcf} output probability distributions, whereas the circuits in \EDlog\ need to output fixed outcomes. We will fix this by making the error, $e$, be part of the input. We cannot, however, make it directly part of the input since the input to the circuits in \EDlog\ is distributed uniformly, whereas $e$ must be drawn from a truncated Gaussian distribution, $D^m_{\mZ_q,B}$. Instead, we will have as part of the input a string $e_u$ which we will turn into a Gaussian sample using a log-depth circuit denoted \textsc{Gaussify}. Implementing this procedure can be achieved using, for instance, the algorithm of Peikert from~\cite{peikert2010efficient}.
\textsc{Gaussify} satisfies the property that if $e_u \leftarrow_U \mZ_q^{m}$ then $e = \textsc{Gaussify}(e_u)$ is distributed according to the truncated Gaussian distribution $ D^m_{\mZ_q,B}$.

The \EDlog\ circuits we construct will therefore have the form
\begin{equation} \label{eqn:cf}
C_f(b, x, e_u) = Ax + b \cdot u + \textsc{Gaussify}(e_u) \; (mod \; q) 
\end{equation}
\vspace{-0.2in}
\begin{equation} \label{eqn:cg}
C_g(b, x, e_u) = Ax + b \cdot (As + e') + \textsc{Gaussify}(e_u) \; (mod \; q) 
\end{equation}
where $A$, $s$, $e'$ and $u$ are fixed. Essentially, one is given $A$, $u$ and $As + e'$ and one has to construct the above circuits and ensure that they have logarithmic depth. Note that the input length for each circuit is $(m + n) \log(q) + 1$. Following~\cite{mahadev2018classical, brakerski2018} we will assume that $m = \Omega(n \log(q))$, so that what we need to ensure is that the circuit depth is $O(\log(m \log(q)))$.
The \LWE\ assumption states that it should be computationally intractable to distinguish $u$ from $As + e'$, when given $A$. However, as we will show, the above circuits will have different entropies in their outputs (when the inputs are chosen uniformly at random). Intuitively this is because one function is $1$-to-$1$ and the other is \emph{approximately} $2$-to-$1$ and so the two cases could be distinguished if we have the ability to solve \EDlog. This is the essence of the reduction.

Let us take stock of all the operations performed by these circuits and why they are all in \NC:
\begin{enumerate}
\item Addition and multiplication modulo $q$ can be performed in logarithmic depth with respect to the input size (which in this case is $O(\log(q))$)~\cite{wallace1964suggestion}, so that this operation requires only depth $O(\log(\log(q)))$. For vectors in $\mZ^m_q$, component-wise addition can be performed in parallel by increasing the width by a factor of $m$. Thus, the overall depth remains $O(\log(\log(q)))$.
\item The dot-product between two vectors in $\mZ^m_q$ requires depth $O(\log(m \log(q)))$. One first computes the component-wise product of the two vectors. This is the same as component-wise addition and can be performed in $O(\log(\log(q)))$ depth. One then adds together all of the results (modulo $q$) and this can be performed in $O(\log(m \log(q)))$ depth with a divide-and-conquer strategy\footnote{Divide the result vector into two vectors of equal length, recursively compute the sum of their components and add the results. Adding the results requires constant depth and the depth of the recursion tree is logarithmic in the length of the vectors.}.
\item Matrix-vector multiplication with matrices in $\mZ_q^{n \times m}$ and vectors in $\mZ_q^{n}$ can be performed in depth $O(\log(m \log(q)))$. Start by copying the vector $m$ times (one copy for each row in the matrix). This can be done in depth $O(\log(m \log(q)))$ with a divide-and-conquer strategy. Then perform the inner products between each row of the input matrix and a corresponding copy of the vector. The inner products can be performed in parallel and each requires $O(\log(m \log(q)))$ depth. Thus, the resulting circuit will have $O(\log(m \log(q)))$ depth.
\item \textsc{Gaussify}\footnote{A simpler, though less efficient way to perform the Gaussian sampling is to take averages of many random numbers (within a range set by the width of the Gaussian) using the uniform randomness and rely on the \emph{central limit theorem}.} can be performed in depth $O(\log(m \log(q)))$ as shown in~\cite{peikert2010efficient}. The procedure from~\cite{peikert2010efficient} requires a pre-processing step of $O(m^3 \log^2(q))$ operations. This will be done as part of the polynomial-time reduction that generates the \EDlog\ instance so that the circuits $C_f$ and $C_g$ already contain the results of this pre-processing step. The actual sampling procedure requires $O(m^2)$ multiplications and additions modulo $q$ which can be performed in parallel requiring depth $O(\log(\log(q)))$. Collecting the results will then require depth at most $O(\log(m \log(q)))$.
\end{enumerate}

This shows that $C_f, C_g \in \NC$. We now estimate the entropies $S(C_f)$ and $S(C_g)$ when the inputs of the circuits are chosen uniformly at random. We will consider $A$ to be a matrix of full rank\footnote{For \LWE, the matrix $A$ is chosen uniformly at random from $\mZ^{n \times m}_q$ and so will be full rank, with high probability.}, $n$. Given that $x \leftarrow_U \mZ^n_q$ and since $A$ is full rank, we have that $Ax \leftarrow_U Im(A)$, where $Im(A)$ denotes the image of $A$. Note that $|Im(A)| = q^n$. For $C_f$, we can choose $u$ such that the distributions $Ax + u$ and $Ax$ have no overlap\footnote{Similar to the choice of $A$, this will be true for most choices of $u$.}. Thus, $Ax + b \cdot u$ will be uniform over $\{0, 1\} \times Im(A)$ and have $n \cdot \log(q) + 1$ bits of entropy. Lastly, we need to account for the Gaussian error. Since this term appears in both $C_f$ and $C_g$, we will simply denote its contribution as $S_{Gaussian}$. We therefore have that $S(C_f) = n \cdot \log(q) + 1 + S_{Gaussian}$.

For the case of $C_g$ the difference will be due to the term $b \cdot (As + e')$. Note that this leads to overlap among different inputs. Specifically $C_g(0, x, e_u) = C_g(1, x - s, e'_u)$, where $e'_u$ is such that $\textsc{Gaussify}(e'_u) = \textsc{Gaussify}(e_u) - e'$. This condition on $e'_u$ is true on all but a negligible fraction of error vectors (as a result of taking $B'/B$ to be super-polynomial in $n$)~\cite{mahadev2018classical, brakerski2018}. The circuit $C_g$ will therefore behave like a $2$-to-$1$ function on all but a negligible fraction of the input domain. We therefore have that $S(C_g) = n \cdot \log(q) + S_{Gaussian} + \mu(n)$, where $\mu(n)$ is a negligible function. For sufficiently large $n$, the $\mu(n)$ term will be less than $1$ and we therefore have that the entropy difference between $C_f$ and $C_g$ is at least a positive constant, as desired\footnote{The entropy difference can be made larger by simply repeating the circuits in parallel. However, an alternate approach presented by the use of ETCF functions is to instead consider the circuits:
\begin{align}
C_f(b_1, b_2, x, e_u) &= Ax + b_1 \cdot u_1 + b_2 \cdot u_2 + \textsc{Gaussify}(e_u) \; (mod \; q)  \\
C_g(b_1, b_2, x, e_u) &= Ax + b_1 \cdot (As_1 + e'_1) + b_2 \cdot (As_2 + e'_2) + \textsc{Gaussify}(e_u) \; (mod \; q)
\end{align}
Here $f$ is still a $1$-to-$1$ function, however $g$ is $4$-to-$1$ so that the entropy difference for these circuits will be $2 - negl(n)$. This construction can be generalised so that, for any constant $k$, the function $g$ can be made into a $2^k$-to-$1$ function.
}.

To complete the proof, note that the reduction we have just described constructs circuits with different entropies starting from an ETCF family (and specifically starting from an \LWE\ instance $(A, As + e')$ and a uniformly random vector $u$). However, being able to distinguish between the output entropies of the circuits allows us to determine which of the two functions is $1$-to-$1$ and which is $2$-to-$1$. By the injective invariance property of ETCF functions (Definition~\ref{def:etcfs}), this is as hard as \LWE, concluding the proof.
\end{proof}
It is worth mentioning that in the above proof we essentially picked ``worst-case'' instances of $A$, $s$, $e'$ and $u$. However, all of the above arguments hold with high probability when $A \leftarrow_U \mZ_q^{n \times m}$, $s, u \leftarrow_U \mZ_q^{n}$ and $e' \leftarrow_{D^m_{\mZ_q,B'}} \mZ_q^{m}$~\cite{mahadev2018classical, brakerski2018}. This means that
\begin{corollary}
\EDlog\ is hard-on-average, based on \LWE\, when the input circuits have the structure given in Equations~\ref{eqn:cf},~\ref{eqn:cg} and are chosen at random according to $A \leftarrow_U \mZ_q^{n \times m}$, $s, u \leftarrow_U \mZ_q^{n}$ and $e' \leftarrow_{D^m_{\mZ_q,B'}} \mZ_q^{m}$.
\end{corollary}

In the previous proof we saw that the ETCF functions based on \LWE\ can be evaluated by circuits of logarithmic depth. It seems unlikely, however, that the same functions could be evaluated by circuits of constant depth. To get around this issue, we make use of the compiling techniques from~\cite{aik} that can take a one-way function with log-depth circuit complexity and map it to a corresponding one-way function having constant-depth circuit complexity. This is achieved through the use of a randomized polynomial encoding, in which the value of a function is encoded in a series of points that can be computed using a constant-depth circuit. Formally, we have that,

\begin{theorem} \label{thm:lweedconst}
\lwe\ $\leq_P$ \EDconst.
\end{theorem}
\begin{proof}
As shown in Theorem~\ref{thm:lweedlog}, the circuits $C_f$ and $C_g$ constructed from the ETCF functions, can be evaluated in log depth. Using Theorem~\ref{thm:pren}, from~\cite{aik}, this means that there exist randomized encodings $\hat{C}_f$ and $\hat{C}_g$ for the two circuits, that can be computed in $\mathsf{NC}^0_4$. To prove our reduction we need to show two things: that the randomized encodings of $C_f$ and $C_g$ preserve the injective invariance property (i.e. distinguishing between the randomized encodings is as hard as \LWE); that the randomized encodings are $1$-to-$1$ and $2$-to-$1$ functions, respectively. This last condition is required so that when we evaluate $\hat{C}_f$ and $\hat{C}_g$ with uniform inputs, it is still the case that $\hat{C}_f$ has more entropy in its output than $\hat{C}_g$.

Showing that injective invariance is satisfied is immediate. Suppose, for the sake of contradiction, that there exists an algorithm that has non-negligible advantage in distinguishing $\hat{C}_f$ and $\hat{C}_g$. It is easy to see that this leads to an efficient algorithm for distinguishing $C_f$ and $C_g$ with non-negligible advantage. Given instances of $C_f$ and $C_g$ we can construct the randomized encodings $\hat{C}_f$ and $\hat{C}_g$. This can be done in polynomial time, according to Theorem~\ref{thm:pren}. We then use our distinguisher on the randomized encodings and this then allows us to distinguish between $C_f$ and $C_g$ which contradicts the injective invariance property.

We now show that the randomized encodings are $1$-to-$1$ and $2$-to-$1$, respectively. Recall first that the randomized encodings take two arguments, $x$ and $r$. First, from Lemma~\ref{lemma:uniquerand} we know that for any fixed input $x$, the encodings are injective in the second argument. In addition, the perfect correctness property (from Definition~\ref{def:randenc}) guarantees that there are simulators $S_f$ and $S_g$ such that $S_f(\hat{C}_{f}(x, r)) = C_f(x)$ and $S_g(\hat{C}_{g}(x, r)) = C_g(x)$, for all $x$ and $r$. This ensures that if $C_f$ is injective then $\hat{C}_f$ will also be injective. It also ensures that whenever there exists a collision in $C_g$, i.e. $x_1$, $x_2$ such that $C_g(x_1) = C_g(x_2)$, there will be a corresponding collision for $\hat{C}_g$, i.e. $\hat{C}_g(x_1, r_1) = \hat{C}_g(x_2, r_2)$, \emph{for all} $r_1$ and $r_2$.
It follows that the randomized encodings $\hat{C}_f$ and $\hat{C}_g$ will have the same entropy difference as $C_f$ and $C_g$, concluding the proof.
\end{proof}

\noindent Just as with the log-depth case, we also have:
\begin{corollary}
\EDconst\ is hard-on-average, based on \LWE\, when the input circuits have the structure given in Equations~\ref{eqn:cf},~\ref{eqn:cg} and are chosen at random according to $A \leftarrow_U \mZ_q^{n \times m}$, $s, u \leftarrow_U \mZ_q^{n}$ and $e' \leftarrow_{D^m_{\mZ_q,B'}} \mZ_q^{m}$.
\end{corollary}
\begin{proof}
The argument is the same as for the log-depth case: for most choices of the circuit parameters (i.e. the matrix $A$, the vectors $s$ and $u$ and the error vector $e'$), we obtain instances of the circuits that satisfy the injective invariance property. As the above proof shows, this remains true for the randomized encodings of these functions as well.
\end{proof}

In the above proofs, we didn't explicitly make use of the fact that the circuits are reversible, though the same results hold in those cases as well (provided we trace out the ancilla required to performed the reversible gates).
Since classical reversible circuits are a particular kind of quantum circuits, these results have the corollary that:

\begin{corollary}
\lwe\ $\leq_P$ \QEDconst.
\end{corollary}
\begin{proof}
Follows from Theorem~\ref{thm:lweedconst} together with the fact that \EDconst\ $\leq_P$ \QEDconst.
\end{proof}
Note, though, that since for the quantum case we are using the same randomized-encoding construction as in the classical case, the resulting quantum circuits must have gates of unbounded fan-out. That is, the quantum circuits considered here belong to the class $\mathsf{QNC^0_f}$. As shown in~\cite{hoyer2005quantum}, constant-depth quantum circuits with unbounded fan-out are very powerful, as they can produce GHZ states (which normally require logarithmic depth to create, when considering gates of bounded fan-out~\cite{watts2019exponential}) and even perform an approximate version of the Quantum Fourier Transform. It would therefore be interesting to see whether the results can be extended to the case of quantum circuits with bounded fan-out gates. We leave this as an interesting open problem.

\section{Hamiltonian quantum entropy difference}
In this section we consider a Hamiltonian analogue of \QED\ which we call Hamiltonian Quantum Entropy Difference (\HQED). The problem will be to estimate the \emph{entanglement entropy} difference between the ground states of two local Hamiltonians. Equivalently, if we trace out parts of the ground states and examine the resulting reduced states, we want to know which of the two has higher Von Neumann entropy. Formally:

\begin{definition}[\HQED]
Let $H_1$ and $H_2$ be local Hamiltonians acting on $n + k$ qubits, whose ground states are $\ket{\psi_{1}}$ and $\ket{\psi_{2}}$. Define the following $n$-qubit mixed states:
\begin{equation}
\rho_1 = Tr_{k} ( \ket{\psi_1} \bra{\psi_1} ) \quad \quad \quad \quad
\rho_2 = Tr_{k} ( \ket{\psi_2} \bra{\psi_2} )
\end{equation}
Given $n$, $k$ and descriptions of $H_1$ and $H_2$ as input, decide whether:
\begin{equation}
S(\rho_1) \geq S(\rho_2) + 1
\end{equation}
or
\begin{equation}
S(\rho_2) \geq S(\rho_1) + 1
\end{equation}
promised that one of these is the case. For the cases where either of the two Hamiltonians has a degenerate groundspace, $\ket{\psi_j}$ will denote a state in the groundspace for which $S(\rho_j)$ is minimal.
\end{definition}

We will refer to \HQEDlog  and \HQEDconst, respectively, as instances of \HQED\ in which the input Hamiltonians have the additional promise that purifications of the states $\rho_1$ and $\rho_2$ can be approximated (to within a $1/poly(n+k)$ additive error in trace distance) by quantum circuits of logarithmic and constant depth, respectively.

\begin{theorem} \label{thm:hqed}
There exists a deterministic poly-time reduction from \QED\ to \HQED\ ($\QED \leq_P \HQED$).
\end{theorem}
\begin{proof}
We could start the reduction by constructing Hamiltonians $H_{1}$ and $H_2$ such that the respective ground states $\ket{\psi_{1}}$ and $\ket{\psi_{2}}$ are Feynman-Kitaev history states for quantum circuits $C_1$ and $C_2$ respectively,
\begin{equation}
\ket{\psi_{j}}=\frac{1}{\sqrt{T+1}}\sum_{t=0}^{T}U_{t}U_{t-1}...U_{1}\ket{00...0}\ket{t}
\end{equation}
where $C_{j}=U_{T}U_{T-1}...U_{1}$, $j \in \{1,2\}$ and the states $\ket{t}$ are the \emph{clock states} in the Feyman-Kitaev history state construction.
However, \textit{a priori} this does not guarantee that determining the entropy difference between the reduced states of $\ket{\psi_{1}}$ and $\ket{\psi_{2}}$ implies determining the entropy difference between the reduced states created by $C_{1}$ and $C_{2}$ respectively. 
This is because the output state of $C_j$ constitutes only one term in the history state superposition. Furthermore, we have no information about the entropy difference between the other terms in $\ket{\psi_1}$ relative to their counterparts in $\ket{\psi_2}$.
To resolve this, we will use a trick from \cite{nirkhe_et_al:LIPIcs:2018:9095}, where a circuit can be padded at the end with identities to give more weight to the ``final term" in the Feynman-Kitaev history state. We will now explain this construction. 

\begin{figure}[h!]
\begin{center}
  \includegraphics[width=0.5\linewidth]{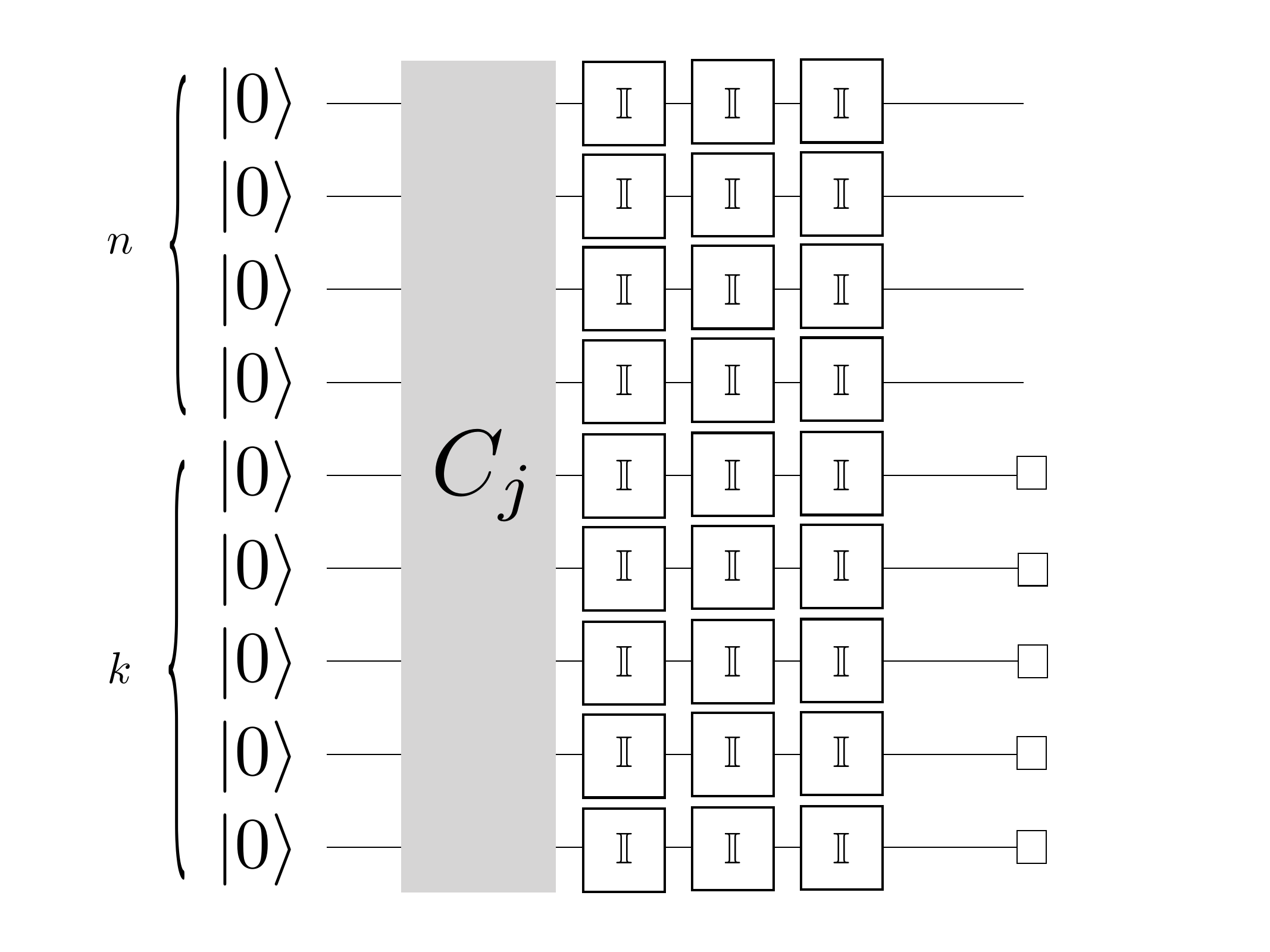}
\caption{Circuit $C_j$ padded with identities.}
\label{fig:paddedcircuit}
\end{center}
\end{figure}

Given the input circuit $C_{j}$ to the problem \QED , first apply $N$ identity operators at the end of the circuit to each qubit, as depicted in Figure \ref{fig:paddedcircuit}; we will call this the padded circuit. Clearly this does not affect the final state, but it will be useful in the reduction to \HQED. We now construct the Feynman-Kitaev history state from the padded circuit, which is
\begin{equation}
\ket{\psi_{j}}=\frac{1}{\sqrt{T+N+1}}\sum_{t=0}^{T+N}U_{t}U_{t-1}...U_{1}\ket{00...0}\ket{t}
\end{equation}
Note that all unitaries $U_i$ for $T+1\leq i \leq T+N$, $U_i=\mathbb{I}^{\otimes (n+k)}$, thus we can simplify the state $\ket{\psi_{j}}$ to be
\begin{equation}
\ket{\psi_{j}}=\frac{1}{\sqrt{T+N+1}}\left((N+1)U_{T}U_{T-1}...U_{1}\ket{00...0}\sum_{t'=T}^{T+N}\ket{t'}+\sum_{t=0}^{T-1}U_{t}U_{t-1}...U_{1}\ket{00...0}\ket{t}\right).
\end{equation}
By the Feynman-Kitaev construction there is a local Hamiltonian, having $poly(n + k)$ terms, for which this is a ground state. This Hamiltonian will act on $n+k+N+T$ qubits due to the clock states $\ket{t}$, which are encoded in unary. Let us now examine the reduced states obtained by tracing out the clock register, which we denote as $\sigma_j$.
In addition, to simplify the notation, we also denote $\ket{\phi_j(t)} = U_{t}U_{t-1}...U_{1} \ket{00...0}$ (with $\ket{\phi_j(0)} = \ket{00...0}$), so that the padded history state can be written as
\begin{equation}
\ket{\psi_{j}}=\frac{1}{\sqrt{T+N+1}}\left((N+1) \ket{\phi_j(T)} \sum_{t'=T}^{T+N}\ket{t'}+\sum_{t=0}^{T-1}\ket{\phi_j(t)}\ket{t}\right).
\end{equation}
Note that $\ket{\phi_j(T)}$ is the output state of the circuit $C_j$. If we now trace out the clock register, we have
\begin{equation}
\sigma_j = Tr_{t}(\ket{\psi_j}\bra{\psi_j}) = \frac{N+1}{T+N+1} \ket{\phi_j(T)}\bra{\phi_j(T)} + \frac{1}{T+N+1} \sum_{t=0}^{T-1}\ket{\phi_j(t)} \bra{\phi_j(t)}
\end{equation}
Computing the fidelity between $\rho_j$ and $\ket{\phi_j(T)}$ we get
\begin{equation}
F(\sigma_j, \ket{\phi_j(T)}\bra{\phi_j(T)}) = \bra{\phi_j(T)} \sigma_j \ket{\phi_j(T)} = \frac{N+1}{T+N+1}  + \frac{1}{T+N+1} \sum_{t=0}^{T-1} |\braket{\phi_j(t) | \phi_j(T)}|^2
\end{equation}
hence
\begin{equation}
F(\sigma_j, \ket{\phi_j(T)}\bra{\phi_j(T)}) \geq \frac{N+1}{T+N+1} = 1 - \frac{T}{T + N + 1}.
\end{equation}
By taking $N = poly(T)$, and given that $T = poly(n+k)$, we get that
\begin{equation}
F(\sigma_j, \ket{\phi_j(T)}\bra{\phi_j(T)}) \geq 1 - \frac{1}{poly(n+k)}.
\end{equation}
From the relationship between fidelity and trace distance this also means that
\begin{equation}
TD(\sigma_j, \ket{\phi_j(T)}\bra{\phi_j(T)}) \leq \frac{1}{poly(n+k)}.
\end{equation}
If we now trace out the $k$ qubits from both of these states and use the fact that this is a trace non-increasing operation, we get
\begin{equation}
TD(\rho'_j, \rho_j) \leq \frac{1}{poly(n+k)},
\end{equation}
where $\rho_j$ is the output state of $C_1$ when tracing out the subsystem of $k$ qubits and $\rho'_j$ is the analogous state for the Hamiltonian $H_j$.
Next, we apply the Fannes-Audanaert inequality~\cite{fannes1973continuity, audenaert2007sharp} relating trace distance and entropy, which says that if
\begin{equation}
TD(\rho'_j, \rho_j) \leq \epsilon
\end{equation}
then
\begin{equation}
| S(\rho'_j) - S(\rho_j) | \leq \frac{\epsilon}{2} (n-1) + h \left( \frac{\epsilon}{2} \right),
\end{equation}
where $h$ is the binary entropy function. Given that $\epsilon = 1/poly(n+k)$, it follows that
\begin{equation}
| S(\rho'_j) - S(\rho_j) | \leq \frac{1}{poly(n+k)}.
\end{equation}
By the triangle inequality we get that if $S(\rho_1) \geq S(\rho_2) + O(1)$, then $S(\rho'_1) \geq S(\rho'_2) + O(1)$ and if $S(\rho_2) \geq S(\rho_1) + O(1)$, then $S(\rho'_2) \geq S(\rho'_1) + O(1)$.

Thus, from the instance of \QED\ $(n + k, k, C_1, C_2)$ we have constructed an instance of \HQED\ $(n + k + N + T, k + N + T, H_1, H_2)$ that preserves the entropy difference of the original circuits (up to a $1/poly(n+k)$ error). This concludes the proof.
\end{proof}

\begin{corollary}
\QEDlog\ $\leq_P$ \HQEDlog , $QED_{O(1)}^b \leq_P$ \HQEDconst.
\end{corollary}
\begin{proof}
In the proof of Theorem~\ref{thm:hqed}, we constructed Hamiltonians $H_1$ and $H_2$ for which the reduced ground states, $\rho'_1$ and $\rho'_2$, are close in trace distance to the output states of the \QED\ circuits $C_1$ and $C_2$. Furthermore, the padded construction used in the previous theorem does not alter the depth of the original circuits, since we are padding with identity gates.
Thus, the depth of the circuits required to approximate $\rho'_1$ and $\rho'_2$ (to within additive error $1/poly(n+k)$ in trace distance) is upper bounded by the depth of $C_1$ and $C_2$ (as $\rho_1$ and $\rho_2$ are approximations of these states). 
For the cases in which these circuits are log-depth or constant-depth, respectively, we obtain instances of \HQEDlog\ and \HQEDconst, respectively.
\end{proof}

For the constant depth case, in Appendix~\ref{sect:constdepthham}, we give a different construction for a local Hamiltonian for which the ground state is \emph{exactly} the output of a quantum circuit from $QED_{O(1)}^b$.

\subsection{Applications to holography} \label{subsect:holo}
Holographic duality is an idea inspired from early results of Bekenstein and Hawking showing that the entropy of a black hole is proportional to its area~\cite{bekenstein, hawking}. This connection between a quantum mechanical property, the Von Neumann entropy of a quantum state, and geometry, in the form of the black hole area, was later expanded upon through the AdS/CFT correspondence~\cite{adscft}. Briefly, the AdS/CFT correspondence is a duality between a non-gravitational quantum field theory (the conformal field theory, or CFT) and a quantum gravitational theory that takes place in an Anti-de Sitter (AdS) space-time. 
The CFT is defined on the boundary of the AdS space. The purpose of the correspondence is to be able to relate physical observables in the bulk to observables on the boundary and vice versa through the so-called \emph{AdS dictionary} (or bulk-to-boundary and boundary-to-bulk maps). This would allow for the derivation of predictions in the quantum gravitational bulk theory purely from a non-gravitational boundary theory.

Similar to the Bekenstein-Hawking result, Ryu and Takayanagi showed a correspondence between geometry and entanglement in AdS/CFT~\cite{rt}. This is known as the \emph{Ryu-Takayanagi} formula and it states that, to leading order, the entropy of a state on the boundary CFT is given by the area of a minimal bulk surface that encloses that state.

Since entropy is an important quantity of interest in AdS/CFT, we discuss potential implications of our result for this duality\footnote{We note that \QSZK\ has appeared before in holography in the context of the Harlow-Hayden decoding task~\cite{harlow2013quantum}. Briefly, Harlow and Hayden considered the task of decoding information from the Hawking radiation of a black hole and showed that one could encode a \QSZK-complete problem in this task. The result was later improved by Aaronson who showed that decoding the information from the radiation would allow one to invert general one-way functions~\cite{aaronson2016complexity}.} In particular we propose a gedankenexperiment based on our results that gives evidence for certain instances of AdS/CFT having the AdS dictionary be computationally intractable to compute, unless \LWE\ is tractable.
A similar result was obtained by Bouland et al \cite{boulandfeffermanvazirani}, in the context of the wormhole growth paradox. Their result also uses cryptographic techniques in the form of pseudorandom quantum states. In contrast to our setting, they only require that such states exist and are computationally indistinguishable, whereas we are using the more fine-grained \LWE\ assumption.

Roughly speaking, the main idea here is that the Ryu-Takayanagi formula relates a quantity that we have shown is hard to compute even for shallow circuits (the entropy), to a quantity that seemingly can be efficiently computed, the area of a surface. Thus assuming that \lwe\ is hard for polynomial time quantum computers, we arrive at a contradiction. A potential resolution is that the AdS dictionary does not efficiently translate from the bulk to the boundary. 

In the following thought experiment, we will assume that CFT states can be prepared efficiently starting from descriptions of functions $f$ and $g$, such as the ETCF functions used in Theorem \ref{thm:lweedlog}, that are $1$-to-$1$ and $2$-to-$1$, respectively. Furthermore, it should be the case that there is a constant difference in entanglement entropy for the two types of states\footnote{As alluded to in the proof of Theorem \ref{thm:lweedlog}, we can in fact make the entropy difference be any constant, $k$, by taking $g$ to be a $2^k$-to-$1$ function (and changing $f$ appropriately, though keeping it a $1$-to-$1$ function).}. To give arguments for why we think this is true, first note that, as stated in Theorem~\ref{thmhqed}, we can construct local Hamiltonians for which the ground states will indeed encode instances of such functions. These ground states will have different entanglement entropy depending on which function was used.

A second argument is based on the observation that certain quantum error-correcting codes serve as toy models for the AdS/CFT correspondence~\cite{pastawski2015holographic}. Specifically, as discussed in~\cite{harlow2017ryu}, codes that protect against erasure errors constitute such toy models. They satisfy the property that encoded information can be recovered by acting on only a fraction of the qubits in the encoded state. As an example of this (taken from~\cite{harlow2017ryu}), if we let $\ket{\tilde{\psi}}_{123}$ be an encoding of $\ket{\psi}$ on three subsystems, it should be that there exists a unitary $U_{12}$ such that
\begin{equation} \label{eqn:qeccstate}
\ket{\tilde{\psi}} = U_{12} [ \ket{\psi}_1 \otimes \ket{\chi}_{23}]
\end{equation}
as well as corresponding unitaries $U_{13}$ and $U_{23}$ that act in a similar way.
Here $\ket{\chi}$ is the maximally entangled state $\ket{\chi} = \sum_i \ket{i}\ket{i}$.
As a toy model for AdS/CFT, the state $\ket{\psi}$ represents the quantum state of a bulk observer and $\ket{\tilde{\psi}} $ will be the corresponding CFT state that lives on the boundary of the AdS space. The indices label three different subsystems on the boundary and Equation~\ref{eqn:qeccstate} simply states that the bulk information (the state $\ket{\psi}$) can be recovered by acting on only part of the boundary. As shown in~\cite{harlow2017ryu}, these states satisfy a Ryu-Takayanagi formula (in addition to other properties that are satisfied by AdS/CFT). Specifically, the entanglement entropy of $\ket{\chi}$ corresponds to the area of a bulk surface. One could imagine considering the states
\begin{equation} \label{eqn:chifg}
\ket{\chi_f} = \sum_x \ket{x} \ket{f(x)} \quad \quad \quad \ket{\chi_g} = \sum_x \ket{x} \ket{g(x)}
\end{equation}
instead of $\ket{\chi}$, where $f$ and $g$ are $1$-to-$1$ and $2$-to-$1$, respectively. 
In this case the difference in entanglement entropy will be determined by whether the function that was used was $1$-to-$1$ or $2$-to-$1$. States such as the ones from Equation~\ref{eqn:chifg}, or even analogous weighted superpositions of such states are efficiently preparable (according to the efficient evaluation property from Definition~\ref{def:etcfs}).

Finally, note that since $f$ and $g$ themselves can be implemented by circuits of constant depth, as shown in Theorem~\ref{thm:lweedconst}, the quantum states derived from these functions (such as $\ket{\chi_f}$ or $\ket{\chi_g}$) could also be prepared by short depth quantum circuits. This would be consistent with a conjecture by Swingle that the underlying CFT states that lead to the Ryu-Takayanagi formula are well approximated by states resulting from MERA (\emph{multi-scale entanglement renormalization ansatz})~tensor networks \cite{swingle}. Such MERA states essentially have log-depth quantum circuit descriptions.
Let us now describe our thought experiment.

Suppose Alice has a quantum computer and is given the description of a function, denoted $h$, which is promised to be either $f$ (that is a $1$-to-$1$) or $g$ (that is a $2$-to-$1$) from Equation~\ref{eqn:etcf}. Alice is then asked whether the function she received is $1$-to-$1$ or $2$-to-$1$. By the injective invariance property (Definition~\ref{def:etcfs}) this is as hard to determine as solving \LWE. Suppose now that Alice uses her quantum computer to do the following:
\begin{enumerate}
\item She first prepares a state $\ket{\psi^h_{CFT}}$ that is supposed to represent a CFT state whose entanglement entropy is determined by the type of function of $h$. In other words, if $h$ is $f$, we will say that the state has high entanglement entropy and if $h$ is $g$ we will say it has low entanglement entropy. As discussed above, we conjecture that there should exist an efficient procedure for preparing such a state, given the function description.
\item By the AdS/CFT correspondence, $\ket{\psi^h_{CFT}}$ should be dual to a state $\ket{\psi^h_{bulk}}$ in the bulk. In this bulk space-time, under the Ryu-Takayanagi formula the area of a certain surface $\gamma_h$ will be equal\footnote{Or be approximately equal.} to the entanglement entropy of $\ket{\psi^h_{CFT}}$. Using the AdS dictionary, Alice then considers a bulk Hamiltonian $H_{bulk}$ such that the time evolution of $\ket{\psi^h_{bulk}}$ under $H_{bulk}$ corresponds to an observer in the bulk measuring the area of $\gamma_h$. If this fictional bulk observer notices that the area of $\gamma_h$ is above a certain threshold (corresponding to the case of high entropy), it will ``reset itself'' so that at the end of the evolution it returns to the state $\ket{\psi^h_{bulk}}$ (and so the corresponding CFT state returns to $\ket{\psi^h_{CFT}}$). If, on the other hand the area is below the threshold (corresponding to low entropy) it should then map itself into a state for which the dual CFT state is ``as close to orthogonal to $\ket{\psi^h_{CFT}}$ as possible''. In other words, we would like this state to be distinguishable from $\ket{\psi^h_{CFT}}$. A schematic illustration of the fictional bulk observer's two perspectives is shown in Figure~\ref{fig:adscft}. The time required for the observer to perform the measurement should be proportional to the area\footnote{Note that since we don't know the entanglement entropy of $\ket{\psi^h_{CFT}}$ in advance, we can take the time we evolve by $H_{bulk}$ to be the longest of the two choices, corresponding to the case of maximum entropy.} of $\gamma_h$. 
\item Using the AdS dictionary, Alice computes the boundary CFT Hamiltonian $H_{CFT}$ that is dual to $H_{bulk}$. She then time-evolves her state $\ket{\psi^h_{CFT}}$ with $H_{CFT}$. Under the AdS/CFT correspondence the evolution of her state will be dual to the time-evolution of the bulk observer that is performing the measurement of $\gamma_h$. Alice is, in effect, simulating this process on her quantum computer.
\item At the end of the evolution, Alice performs SWAP tests to check whether the state she is left with is $\ket{\psi^h_{CFT}}$. If this is the case, she concludes that the original function was $1$-to-$1$, otherwise she concludes that it is $2$-to-$1$.
\end{enumerate}

\begin{figure}[ht!]
\begin{center}
  \includegraphics[width=0.8\linewidth]{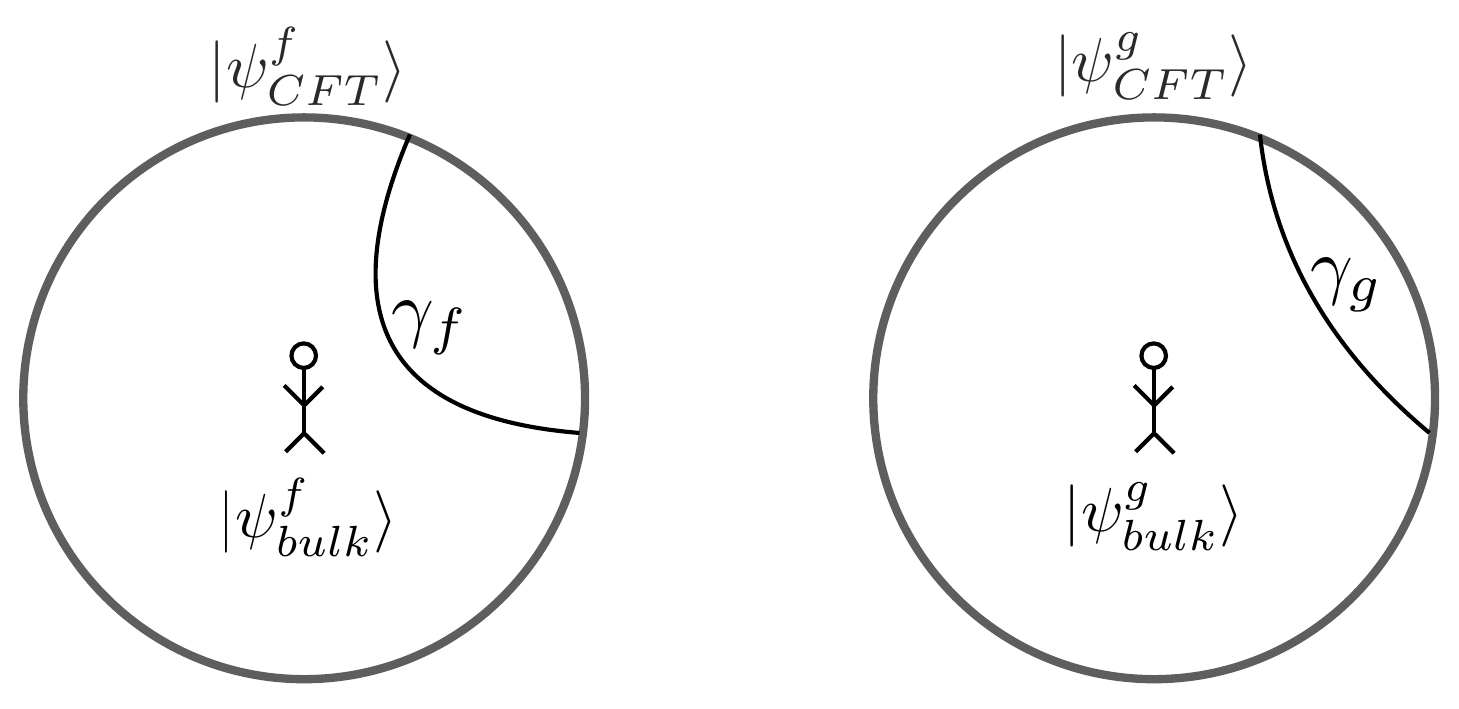}
\caption{The two situations for the bulk observer. The observer should measure the area of the surface $\gamma$ to determine which situation it is in.}
\label{fig:adscft}
\end{center}
\end{figure}

If all the steps in the above procedure can be performed efficiently, then this experiment would violate the injective invariance property of ETCF functions, since it can efficiently distinguish between the two function types. Correspondingly, we would have an efficient quantum algorithm for solving \LWE. Since we believe that this is unlikely, we would need to determine which of the above steps is intractable. As mentioned, we conjecture that preparing the CFT states should be efficient. The time-evolution under a Hamiltonian, that Alice has to perform, should also be efficient using standard techniques from quantum simulation~\cite{brown2010using, berry2015hamiltonian}. One step that seems potentially problematic is step 3. Here the bulk observer needs to affect his space-time so as to map Alice's state to one of two efficiently distinguishable states. It certainly seems plausible that the bulk observer can do very different things depending on the area he measures, resulting in completely different bulk states. But it is unclear whether the resulting dual CFT states would then be distinguishable by Alice. The other possible source of intractability is the use of the AdS dictionary. If the dictionary is exponentially complex, Alice cannot determine the state that is dual to her CFT state or what her boundary Hamiltonian should be.

An important observation to make here is that the entropy difference between the two cases is constant and one could argue that we should not expect bulk observers to be able to efficiently detect such small differences in geometry. Indeed, the Ryu-Takayanagi formula typically states that entropy equals geometry \emph{to leading order}. The entropy of the bondary state will have higher-order corrections stemming from the entanglement-entropy  of states in the bulk. Thus,
it might be more relevant to have a scenario in which the \emph{entropy ratio} is constant instead of the entropy difference, since this would correspond to a noticeable change in area between the two cases (and would be more in line with the portrayal in Figure~\ref{fig:adscft}).
As we show in Appendix~\ref{sect:eratio}, even estimating the entropy ratio is hard based on \LWE\ and, in fact, entropy difference and entropy ratio are computationally equivalent.

A final comment we make about the above experiment is that it does not require holographic duality to be true for our own universe. Indeed, as long as AdS/CFT is true it should in principle be possible to simulate dynamics in a ``virtual'' AdS space by constructing CFT states on a quantum computer and evolving them under the CFT Hamiltonian. 
Can instances of \LWE\ and ETCF functions be encoded in these CFT states? We leave answering this question and formalizing the above experiment for future work.

\bibliography{biblio}
\bibliographystyle{alpha}

\appendix
\section{An unconditional lower bound on entropy estimation} \label{sect:bqpqpoly}
A question we can ask concerning entropy estimation is whether the entropy of a quantum state, $\rho$, on $n$-qubits can be estimated up to additive error $\epsilon$ in time (and with a number of copies of $\rho$) that scales as $poly(n, \log(1/\epsilon))$. While it has already been shown in~\cite{acharya, quantumalgorithmentropy} that the answer is no, here we give a complexity-theoretic proof of this fact:

\begin{theorem}
There is no quantum algorithm running in time $poly(n, \log(1/\epsilon))$ for estimating the entropy of an $n$-qubit quantum state $\rho$ to within additive error $\epsilon > 0$.
\end{theorem}
\begin{proof}
We prove this result by contradiction. Suppose that a $poly(n, \log(1/\epsilon))$ quantum algorithm for entropy estimation existed. We will argue that this implies $\mathsf{BQP/qpoly} = \mathsf{ALL}$, where $\mathsf{BQP/qpoly}$ denotes the set of languages that can be decided by a polynomial-time quantum algorithm with quantum advice, and $\mathsf{ALL}$ denotes the set of all languages~\cite{zoo}.
A $\mathsf{BQP/qpoly}$ algorithm is a quantum algorithm that runs in polynomial time and that receives, in addition to its input, denoted $x$, a quantum state on $poly(|x|)$ qubits (the advice state) that depends only on the size of the input and not on the input itself.
We leverage this fact, together with the ability to efficiently estimate entropy to provide an algorithm for deciding any language.

For a given language $L \subseteq \{0, 1\}^*$ and input length $n$, let $L_n$ denote the set of strings $x$ of length $n$, such that $x \in L$. We also let $\bar{L}_n$ denote the strings of length $n$ not contained in $L$, i.e. $\bar{L}_n = \{0, 1\}^n \setminus L_n$. With this notation, we define the states
\begin{equation}
\ket{\psi_n}_{Yes} = \frac{1}{\sqrt{L_n}} \sum_{x \in L_n} \ket{x}  \quad \quad \quad
\ket{\psi_n}_{No} = \frac{1}{\sqrt{\bar{L}_n}} \sum_{x \in \bar{L}_n} \ket{x}
\end{equation}
to be the equal superpositions over the ``yes'' instances of length $n$ and the ``no'' instances, respectively.
Finally, we let $EQ_x$ be the following unitary operation acting on $n+1$ qubits (for $b\in\{0,1\}$):
\begin{equation}
EQ_x \ket{y} \ket{b} = \left\{
\begin{array}{ll}
      \ket{y}\ket{b \oplus 1} & \text{ if } x=y  \\
      \ket{y}\ket{b} & \text{ otherwise }\\
\end{array} 
\right.
\end{equation}
Note that $EQ_x$ can be implemented with $poly(|x|)$-many gates.

The $\mathsf{BQP/qpoly}$ algorithm works as follows. Setting $\epsilon = 2^{-n-1}$, the quantum advice will consist of $poly(n)$ copies of $\ket{\psi_n}_{Yes}\ket{\psi_n}_{No}$.
The algorithm then appends two qubits in the state $\ket{0}$ to each copy $\ket{\psi_n}_{Yes}\ket{\psi_n}_{No}$ to get $\ket{\psi_n}_{Yes}\ket{0}\ket{\psi_n}_{No}\ket{0}$. Then for the input $x$, the algorithm applies $EQ_x$ to both $\ket{\psi_n}_{Yes}\ket{0}$ and $\ket{\psi_n}_{No}\ket{0}$ individually, for all copies. Consider what happens if $x$ is a ``yes'' instance (the ``no'' instance case is analogous). The resulting states will be
\begin{equation} \label{eqn:yes}
EQ_x \ket{\psi_n}_{Yes} \ket{0} = \frac{1}{\sqrt{L_n}} \sum_{z \in L_n, z \neq x} \ket{z}\ket{0} + \frac{1}{\sqrt{\bar{L}_n}} \ket{x}\ket{1}
\end{equation}
\begin{equation} \label{eqn:no}
EQ_x \ket{\psi_n}_{No} \ket{0} = \ket{\psi_n}_{No} \ket{0}
\end{equation}
If we trace out the first $n$ qubits and denote the resulting states as $\rho_Y$ and $\rho_N$, we can see that
\begin{equation}
S(\rho_Y) = -\frac{L_n - 1}{L_n} \log \left(\frac{L_n - 1}{L_n} \right) - \frac{1}{L_n} \log \left(\frac{1}{L_n} \right)
\end{equation}
\begin{equation}
S(\rho_N) = 0
\end{equation}
Since $L_n \leq 2^n$, we have that $S(\rho_Y) \geq 2^{-n}$. But now, by assumption, having $poly(n)$-many copies of $\rho_Y$ and $\rho_N$ we can estimate the entropies of the two states to within additive error $2^{-n-1}$, thus being able to determine which of the two has non-zero entropy.
The algorithm is therefore able to decide any language $L$, hence showing that $\mathsf{BQP/qpoly} = \mathsf{ALL}$. However, we know from~\cite{ny} that $\mathsf{BQP/qpoly} \neq \mathsf{ALL}$ (since, in particular $\mathsf{EESPACE} \not\subset \mathsf{BQP/qpoly}$) and this provides the desired contradiction.
\end{proof}

\section{The entropy ratio problem} \label{sect:eratio}
We have seen that estimating entropy differences for circuit outputs, even constant depth circuits, is hard. A related question we can ask is: what about entropy \emph{ratio}? In other words, what can we say about the hardness of estimating $S(\rho_1)/S(\rho_2)$, where $\rho_1$ and $\rho_2$ are states produced by circuits $C_1$ and $C_2$, respectively?\footnote{Assuming $S(\rho_2) > 0$.}

The motivation for this question stems from the relationship between entropy and geometry in the AdS/CFT correspondence, which we discuss in more detail in Subsection~\ref{subsect:holo}. Briefly, there can be situations in which the relevant physical quantity to consider is the ratio in entropy between two states (or distributions), rather than the difference. In AdS/CFT this is the case because various geometric quantities (such as lengths, areas or volumes) are equal, to first order, to the entropy of a quantum state. Observers in a spacetime are generally able to determine ratios between these geometric quantities and thus
obtain estimates for the entropy ratio between quantum states. It is therefore important to characterise the hardness of estimating these ratios. We thus define the following problem\footnote{The classical version of the problem is analogous.}:

\begin{definition}[Quantum Entropy Ratio (\QER)]
Let $C_1$ and $C_2$ be quantum circuits acting on $n + k$ qubits. Define the following $n$-qubit mixed states:
\begin{equation}
\rho_1 = Tr_{k} (C_1 \ket{00...0} \bra{00...0} C_1^{\dagger}) \quad \quad \quad \quad
\rho_2 = Tr_{k} (C_2 \ket{00...0} \bra{00...0} C_2^{\dagger})
\end{equation}
Given $n$, $k$, $a$, $b$ with $a, b > 0$, $a - b > 1/poly(n + k)$, as well as descriptions of $C_1$ and $C_2$ as input, decide whether:
\begin{equation}
\frac{S(\rho_1)}{S(\rho_2)} \geq a
\end{equation}
or
\begin{equation}
\frac{S(\rho_1)}{S(\rho_2)} \leq b
\end{equation}
promised that one of these is the case and promised that $S(\rho_2) > 1/poly(n + k)$.
\end{definition}

Before proceeding further, let us first comment on the condition $S(\rho_2) > 1/poly(n + k)$. First of all it is clear that $S(\rho_2) > 0$, however we can see that if the entropy of $\rho_2$ can be arbitrarily close to $0$ the problem becomes ``trivially'' $\mathsf{NP}$-hard\footnote{In fact, a more involved version of this argument can be used to show the problem is also $\mathsf{\#P}$-hard.}. The reason is that one can consider $C_1$ to be a circuit that creates a superposition over all satisfying assignments of some CNF formula, $\phi$. Taking $n = 1$, and $k$ as the number of variables of $\phi$, the circuit then flips the last qubit conditioned on having a satisfying assignment for $\phi$. It's not hard to see that in this case
\begin{equation}
S(\rho_1) = - \frac{\# SAT(\phi)}{2^k} \log \left( \frac{\# SAT(\phi)}{2^k} \right) - \left( 1 - \frac{\# SAT(\phi)}{2^k} \right) \log \left(1 - \frac{\# SAT(\phi)}{2^k} \right) 
\end{equation}
where $\# SAT(\phi)$ denotes the number of satisfying assignments of $\phi$.
Similar to the proof from Appendix~\ref{sect:bqpqpoly}, this means that if $\#SAT(\phi) > 0$, then $S(\rho_1) > 2^{-k}$.
For $C_2$, we can consider a circuit that creates a state, $\rho_2$, for which $S(\rho_2) = 2^{-k}$. It therefore follows that the entropy ratio is $0$ when $\#SAT(\phi) = 0$ and at least\footnote{Note that this is true provided that not all assignments of $\phi$ are satisfying assignments, since in that case $S(\rho_1) = 0$. However, this case can easily be eliminated by taking an additional ``dummy variable'', $x_0$, and considering instead the formula $\phi' = x_0 \wedge \phi$. We can see that $\phi'$ is satisfying whenever $\phi$ is satisfying (by taking $x_0 = 1$), but not all assignments of $\phi'$ are satisfying (because of $x_0 = 0$).} $1$ when $\#SAT(\phi) > 0$. Thus, taking $a = 1$ and $b = 0$, we have a reduction from SAT to \QER.
To avoid this and have a physically meaningful version of \QER, we require $S(\rho_2) > 1/poly(n + k)$.

What can we say about the hardness of \QER? One way to answer this question is to relate it to \QED. To that end, we show the following:

\begin{lemma} \label{lemma:qerqed}
\QER $\leq_P^{C}$ \QED, where $\leq_P^{C}$ denotes a polynomial-time Cook reduction. Equivalently, \QER $\in \mathsf{BPP}^{\QED}$.
\end{lemma}
\begin{proof}
Supposing $S(\rho_2) \geq 1/p(n + k)$, for some polynomial $p$, take $q(x) = (a - b) p^2(x)$, where $a$ and $b$ are the parameters from \QER. A $\mathsf{BPP}$ algorithm with an oracle for \QED\ can now estimate $S(\rho_1)$ and $S(\rho_2)$ to within precision $1/q(n + k)$, using calls to the oracle. Note that this can be done by first calling the oracle with $C_1$ and the identity circuit and $C_2$ and the identity circuit, respectively. This is because the identity circuit produces a state having zero entropy and so using \QED\ this way, one obtains entropy estimates for $S(\rho_1)$ and $S(\rho_2)$, which we denote as $\tilde{S}_1$ and $\tilde{S}_2$, respectively.
It is the case that 
\begin{equation}
\tilde{S}_i = S(\rho_i) + \epsilon_i,
\end{equation}
for some $\epsilon_i$ such that $|\epsilon_i| \leq 1/q(n + k)$ and $i \in \{1, 2\}$. 

Finally, the $\mathsf{BPP}$ algorithm computes $\tilde{S}_1 / \tilde{S}_2$, or
\begin{equation}
\frac{\tilde{S}_1}{\tilde{S}_2} = \frac{S(\rho_1) + \epsilon_1}{S(\rho_2) + \epsilon_2}. 
\end{equation}
But now since $S(\rho_2) > \epsilon_2$, by Taylor expanding the fraction to first order in $\epsilon_2$, we have that
\begin{equation}
\frac{S(\rho_1) + \epsilon_1}{S(\rho_2) + \epsilon_2} = \frac{S(\rho_1)}{S(\rho_2)} + O(\epsilon_2) 
\end{equation}
Because $|\epsilon_2| \leq 1/q(n + k)$, the estimate of the entropy ratio is good enought to decide whether 
\begin{equation}
\frac{S(\rho_1)}{S(\rho_2)} > a
\end{equation} 
or
\begin{equation}
\frac{S(\rho_1)}{S(\rho_2)} < b
\end{equation}
concluding the proof.
\end{proof}

The above result indicates that estimating the entropy ratio (when we are promised that one of the circuits produces a state that has a non-trivial amount of entropy) is easier than estimating entropy difference. However, we can perform a similar reduction in the other reduction showing that, in fact, the two problems are equivalent (under Cook reductions):

\begin{lemma} \label{lemma:qedqer}
\QED $\leq_P^{C}$ \QER, where $\leq_P^{C}$ denotes a polynomial-time Cook reduction. Equivalently, \QED $\in \mathsf{BPP}^{\QER}$.
\end{lemma}
\begin{proof}
Analogously to Lemma~\ref{lemma:qerqed}, we will use the oracle to \QER\ in order to estimate the entropies $S(\rho_1)$ and $S(\rho_2)$, corresponding to the two circuits, and then take their difference. Assuming their difference is lower bounded by $1/p(n + k)$, for some polynomial $p$, we can again take $q(x) = p^2(x)$, and perform this estimation to within precision $1/q(n + k)$. The key to doing this will be to define a new circuit, $C$, which simply creates one Bell pair. If one of the qubits in this pair is traced out, the remaining qubits will have entropy exactly $1$.

Thus, the $\mathsf{BPP}$ algorithm will first use the oracle to \QER\ with circuits $C_1$ and $C$ and then with $C_2$ and $C$. Since the entropy of the state produced by $C$ is $1$, this procedure will yield estimates for the entropies $S(\rho_1)$ and $S(\rho_2)$. As their difference is assumed to be at least $1/p(n + k)$ and since the estimates are within precision $1/p^2(n+k)$, taking their difference will be sufficient to determine whether
\begin{equation}
S(\rho_1) \geq S(\rho_2) + 1/p(n+k)
\end{equation}
or
\begin{equation}
S(\rho_2) \geq S(\rho_1) + 1/p(n+k)
\end{equation}
concluding the proof.
\end{proof}

\begin{corollary}
\QED\ and \QER\ are equivalent under Cook reductions.
\end{corollary}
\begin{proof}
Follows from Lemmas~\ref{lemma:qerqed}, \ref{lemma:qedqer}.
\end{proof}

The above results remain valid even for circuits of constant depth. Of course, in that case it would be desirable to have Karp reductions, or even log-space reductions between the problems. While we leave finding such reductions as an open problem, we can at least show that such reductions do exist from \LWE\ to the constant-depth version of \QER:

\begin{lemma}
\lwe\ $\leq_P$ \QERconst. 
\end{lemma}
\begin{proof}
This reduction follows exactly the same steps as Theorems~\ref{thm:lweedlog},~\ref{thm:lweedconst}, except that instead of using a function $g$ that is $2$-to-$1$, we use a function that is $2^r$-to-$1$, for some $r > 0$. Specifically, the functions we consider are:
\begin{equation}
f(b_1, b_2, ... b_r, x) = Ax + \sum_{i=1}^r b_i u_i + e \; (mod \; q)
\end{equation}
\begin{equation}
g(b_1, b_2, ... b_r, x) = Ax + \sum_{i= 1}^r b_i (As_i + e_i) + e \; (mod \; q)
\end{equation}
with $q \geq 2$ a prime integer, $i \leq r$, $b_i \in \{0, 1\}$, $x \in \mZ_q^n$, $s_i \in \mZ_q^n$, $A \in \mZ^{n \times m}_q$, $u_i, e_i \in \mZ_q^m$, $e \leftarrow_{ D_{\mZ_q,B}^m} \mZ^m$.

We can see that $f$ is indeed (approximately) a $1$-to-$1$ function, while $g$ is (approximately) $2^r$-to-1, since, for any $x \in \mZ_q^n$,
\begin{equation}
g(0, 0, ... 0, x) \approx g \left( b_1, b_2, ... b_r, x - \sum_{i= 1}^r b_i s_i \right)
\end{equation}
for all $b_i \in \{0, 1\}$, $i \leq r$.
Moreover, it is also the case that $f$ and $g$ can be implemented given the uniformly random vectors $u_i$ and the LWE samples $As_i + e_i$, respectively. So long as $r = poly(n)$, the vectors $u_i$ and the LWE samples $As_i + e_i$ are computationally indistinguishable from each other~\cite{lwe}.
The circuits implementing these two functions will again be denoted as $C_f$ and $C_g$, as in Theorem~\ref{thm:lweedlog}, and they are similarly:
\begin{equation}
C_f(b_1, b_2, ... b_r, x, e_u) = Ax + \sum_{i=1}^r b_i u_i + \textsc{Gaussify}(e_u) \; (mod \; q) 
\end{equation}
\vspace{-0.2in}
\begin{equation}
C_g(b_1, b_2, ... b_r, x, e_u) = Ax + \sum_{i= 1}^r b_i (As_i + e_i) + \textsc{Gaussify}(e_u) \; (mod \; q) 
\end{equation}
The only difference between these circuits and the ones from Theorem~\ref{thm:lweedlog}, is the computation of the sums $\sum_{i=1}^r b_i u_i$ and $\sum_{i= 1}^r b_i (As_i + e_i)$, respectively. However, these sums can be computed in parallel using $O(\log(r \log(q)))$ depth, by recursively evaluating the first half of the sum and the second half separately and then combining the results. Thus, as long as $r = poly(n)$, $C_f, C_g \in \NC$.

Finally, let us discuss the entropies we obtain for the two circuits under uniform inputs. We can see that in the case of $f$, each $b_i$ is adding an additional bit of entropy to that obtained from having a uniform $x$ input. Hence $S(C_f) = r + n \log(q) + S_{Gaussian} + \mu(n)$, for some negligible function $\mu(n)$. For $C_g$ we have $S(C_g) = n \log(q) + S_{Gaussian} + \mu(n)$. Also note that $S_{Gaussian} \leq n \log(q)$ so that by taking $r = p(n) (2 n \log(q) + \mu(n))$, for some polynomial $p$, we have that
\begin{equation}
S(C_f) / S(C_g) > p(n)
\end{equation}
and also
\begin{equation}
S(C_g) / S(C_f) < 1/p(n)
\end{equation}

This essentially concludes the proof. We have shown how one can construct log-depth circuits whose entropy ratio (under a uniform input) will either be very large or very small. Moreover, distinguishing which is the case is as hard as \lwe. We can then use randomized encodings as in Theorem~\ref{thm:lweedconst}, to achieve the same thing for constant depth circuits. 
By implementing these as quantum circuits, the results extend to the quantum case as well.
\end{proof}

\section{HQED with constant depth ground state} \label{sect:constdepthham}
In the reduction from Theorem~\ref{thm:hqed} we used the history state construction to \emph{approximately} map the outputs of circuits $C_1$ and $C_2$ from an instance of \QED\ to ground states of Hamiltonians $H_1$ and $H_2$ in \HQED.
Correspondingly, the entanglement entropy difference for the ground states of $H_1$ and $H_2$ differed from that of the output states of $C_1$ and $C_2$ by an additive term of $1/poly(n+k)$.

Here we give an alternate reduction for the case where $C_1$ and $C_2$ are of constant depth $d$, based on a recent result of Bravyi, Gosset and Movassagh~\cite{quantummeanvalues}.
This reduction has the appealing feature that the resulting Hamiltonians will have as ground states \emph{exactly} $C_1 \ket{00...0}$ and $C_2 \ket{00...0}$, rather than approximate versions of these states. This means that the entanglement entropy difference for the ground states of $H_1$ and $H_2$ will be identical to that of the states produced by $C_1$ and $C_2$.

\begin{lemma}
There exists a reduction \QEDconst\ $\leq_P$ \HQEDconst\ that exactly preserves the entropy difference.
\end{lemma}
\begin{proof}
Assuming, as before, that the circuits acting on $n + k$ qubits have the form $C_j = U_{d} U_{d-1} ... U_1$, $j \in \{1, 2\}$, where each $U_i$ is a quantum gate and $d$ is constant, we use the Hamiltonians considered in~\cite{quantummeanvalues}:
\begin{equation}
H_j = \sum_{i = 1}^{n + k} C_j \ket{1}\bra{1}_i C_j^{\dagger}
\end{equation}
where $\ket{1}\bra{1}_i$ acts non-trivially only on the $i$'th qubit. Since $C_j$ is a circuit of depth $d$, the locality of each term is $2^d$. Because $d$ is constant, the resulting Hamiltonian is local.
As shown in~\cite{quantummeanvalues}, the unique ground state of $H_j$ is $C_j \ket{00...0}$.
Thus, the entanglement entropy difference for $H_1$ and $H_2$ is given by the entanglement entropy difference of $C_1 \ket{00...0}$ and $C_2 \ket{00...0}$, concluding the proof.
\end{proof}

\end{document}